\documentclass[a4paper,preprint, pre]{revtex4}

\usepackage{amsmath,amssymb}
\usepackage{makeidx}
\usepackage{amsfonts}
\usepackage[ansinew]{inputenc}
\usepackage[usenames,dvipsnames]{pstricks}
\usepackage{epsfig,graphicx}
\usepackage{color}

\DeclareMathOperator*{\sgn}{sgn}
\newcommand{\vecSpin}{\mbox{\boldmath$S$}}
\newcommand{\vecSpinhat}{\hat\vecSpin}
\newcommand{\vecPsi}{\mbox{\boldmath$\psi$}}
\newcommand{\vecPsihat}{\hat\vecPsi}

\newcommand{\Spin}{\mbox{$S$}}

\newcommand{\Prob}{\mbox{$P$}}

\newcommand{\Psihat}{\hat\psi}
\newcommand{\nullvec}{\mbox{\boldmath$0$}}
\newcommand{\rme}{\mathrm{e}}
\newcommand{\rmi}{\mathrm{i}}
\newcommand{\rmd}{\mathrm{d}}

\newcommand{\setI}{\textbf{I}}

\newcommand{\setN}{[N]}

\newcommand{\1}{\mathbf{1}}

\newtheorem{theorem}{Theorem}[section]

\newtheorem{proposition}[theorem]{Proposition}

\newenvironment{proof}[1][Proof]{\begin{trivlist}
\item[\hskip \labelsep {\bfseries #1}]}{\end{trivlist}}

\newcommand{\qed}{\nobreak \ifvmode \relax \else
      \ifdim\lastskip<1.5em \hskip-\lastskip
      \hskip1.5em plus0em minus0.5em \fi \nobreak
      \vrule height0.75em width0.5em depth0.25em\fi}

\makeindex

\begin{document}

\title{Phase transitions and memory effects in the dynamics of Boolean networks}

\author{Alexander Mozeika}\email{a.s.mozeika@aston.ac.uk}
\author{David Saad}

\affiliation{The Non-linearity and Complexity  Research Group, Aston University, Birmingham B4 7ET, UK.}

\date{\today}

\begin{abstract}

The generating functional method is employed to investigate the synchronous dynamics of Boolean networks, providing an exact result for the system dynamics via a set of macroscopic order parameters. The topology of the networks studied and its constituent Boolean functions represent the system's quenched disorder and are sampled from a given distribution. The framework accommodates a variety of topologies and Boolean function distributions and can be used to study both the noisy and noiseless regimes; it enables one to calculate correlation functions at different times that are inaccessible via commonly used approximations. It is also used to determine conditions for the annealed approximation to be valid, explore phases of the system under different levels of noise and obtain results for models with strong memory effects, where existing approximations break down. Links between BN and general Boolean formulas are identified and common results to both system types are highlighted.

\end{abstract}

\maketitle

\section{Introduction}

Boolean networks (BN) have been suggested as simplified models of various biological systems, in particular for modeling gene-regulatory networks~\cite{Kauffman}. Simplifying the state of genes by adopting a two state variable representation, {\small{\sf ON/OFF}}, enables one to model the complex interactions between genes as arbitrary Boolean functions of randomly selected variables. This model gives rise to a rich and complex behavior that has been successfully employed to gain insight into the dynamics and steady states of various gene-regulatory systems~\cite{RBNbook}.

BN models comprise $N$ sites (genes), each of which is represented by a binary variable. The state of each variable is determined by the states of $k$ randomly selected sites via a $k$-input Boolean function. The specific $k$ input variables selected for each site constitutes the network topology and the selected Boolean functions determine the corresponding interaction leading to the variables' state. In the original formulation~\cite{Kauffman}, both topology and Boolean functions were selected uniformly from the ensemble of networks with in-degree $k$ per variable and $2^{2^k}$ possible Boolean functions, respectively. Both are considered fixed (\emph{quenched} variables). Moreover, the original model was deterministic and the analysis therefore focused on its periodic-orbit attractors, steady-state and their basins of attraction.

While this family of networks was originally introduced to model the gene-regulatory network~\cite{RBNbook}, and is commonly known as Random Boolean Networks (RBN) or Kauffman nets, similar topologies have been employed to study network properties in other application domains, ranging from social~\cite{Moreira} to genetic~\cite{Nature} and neural~\cite{DerridaANN} networks. Although the topology used is common to all these models, based on a discrete state-space and random $k$-variable Boolean interactions, the nature of the interaction may be different for each of the models. We will refer to the general class of $N$-variable binary system with connectivity degree $k$ as the N-k model~\cite{Aldana,Drossel}.

This abstraction of complex gene-regulatory system lends itself to analysis in terms of both their dynamics and equilibrium properties~\cite{Correale:RBN,Leone:RBN}. Equilibrium analysis relies mainly on the cavity method, while the dynamics has been mostly investigated using the \emph{annealed approximation}~\cite{Derrida} due to its simplicity and success in providing accurate results in many of the models studied, especially for very large systems. The underlying approximation in this approach is that both thermal and quenched variables (primarily the network topology) are considered to be on equal footing and are sampled at each time step. This helps suppress emerging correlations of specific sites at different times, simplifies the analysis and gives rise to an effective methodology which works in most cases. The annealed approximation was particularly successful in large-scale systems ($N\rightarrow\infty$); it allows one to predict the evolution of network activity  and Hamming distance order parameters. The former refers to the magnetization or the proportion of {\small{\sf ON/OFF}} states and the latter to the difference between the states of the network starting from different initial conditions. It was shown~\cite{DerridaAndW,Hilhorst} that the annealed approximation provides accurate magnetization and Hamming distance order parameter predictions for RBN (i.e., with uniformly sampled Boolean functions); however, the conditions for its applicability and validity for general N-k systems has remained unclear~\cite{Kesseli}. Moreover, in some cases, especially in systems with memory, discrepancies have been found between results obtained via the annealed approximation and simulation results~\cite{RTN}, casting doubts on the validity of the approach for such models.

The aim of the current paper is to develop further a framework for exact analysis of N-k models based on the generating functional analysis (GFA) framework~\cite{dD,Mozeika:BN}, which has been employed successfully in the study of various Ising spin models~\cite{MimuraAndCoolen}. The newly developed framework is then employed to determine the conditions under which the annealed approximation is valid, to investigate the possible phases of BN depending on the noise level and to demonstrate its efficacy for analyzing systems with memory, where the annealed approximation is known to break down~\cite{RTN}. We note that an alternative to the GFA method (called the \emph{dynamical cavity} method) was recently introduced~\cite{NeriAndBolle}, which we believe can be also used for the range of models studied here.

Section~\ref{section:model} introduces the BN model while Section~\ref{section:GFA} describes the methodology used and its application to the current model. In section~\ref{section:results} we employ the dynamical equations obtained for the order parameters to investigate the conditions under which the annealed approximation provides exact results; a similar set of equations is then used to identify possible phases of the system and to analyze the dynamics of system with memory. Finally, in section~\ref{section:summary}, we summarize the results obtained and point to future research directions. Some of the detailed derivations appear in dedicated appendices.

\section{Model\label{section:model}}
The model we consider here is a recurrent Boolean network which consists of $N$ binary variables $S_{i}(t)\in\{\!-\!1,1\}$ interacting via Boolean functions $\alpha_i:\{-1,1\}^k\rightarrow\{-1,1\}$ of exactly $k$ inputs. Because of thermal noise, which can flip the output of a function with probability $p$~\cite{Noise}, a site $S_{i}(t)$ in the network is operating according to the stochastic rule
\begin{eqnarray}
S_{i}(t\!+\!1)=\eta_{i}(t)\;\alpha_i
(S_{i_1}(t),\ldots,S_{i_k}(t)),\label{def:algorithm}
\end{eqnarray}
where $\eta_{i}(t)$ is an independent random variable from the distribution  $\Prob(\eta)=p\delta_{\eta;-1}\!+\!(1\!-\!p)\delta_{\eta;1}$. The function-output $S_i(t\!+\!1)$ is completely random when $p=1/2$ and completely deterministic when $p=0$. Averaging out the thermal noise $\{\eta_{i}(t)\}$ in the system governed by (\ref{def:algorithm}) gives rise to the microscopic law
\begin{eqnarray}
\Prob_{\alpha_i}(S_{i}(t\!+\!1)\vert S_{i_1}(t),..,S_{i_k}(t))=\frac{\rme^{\beta S_{i}(t\!+\!1)\alpha_i(S_{i_1}(t),..,S_{i_k}(t))}}{2\cosh\beta\alpha_i(S_{i_1}(t),..,S_{i_k}(t))},\label{eq:micro}
\end{eqnarray}
where the inverse temperature $\beta\!=\!1/T$ relates to the noise parameter $p$ via $\tanh\beta\!=\!1\!-\!2p$.

All sites in the network are updated in parallel and given the state of the network $\vecSpin(t)\in\{-1,1\}^N$  at time $t$,  the function-outputs $\vecSpin(t\!+\!1)$  at time $t\!+\!1$ for the different sites are independent of each other. This Markovian property allows us to write the probability of the microscopic path $\vecSpin(0)\rightarrow\!\cdots\!\rightarrow\vecSpin(t_{max})$ as a product of (\ref{eq:micro}) over all sites and time steps. Furthermore, we consider two copies of the same topology but with different initial conditions, shown in Figure \ref{fig:1}, comparing the two will enable us to study the effects of initial-state perturbations. Following similar arguments to those of the single network case, the joint probability of microscopic states in the two systems are given by
\begin{eqnarray}
\Prob[\{\vecSpin(t)\};\{\vecSpinhat(t)\}]\!=\!\Prob(\vecSpin(0),\vecSpinhat(0)) \prod_{t=0}^{t_{max}-1} \Prob(\vecSpin(t\!+\!1)\vert\vecSpin(t))P(\vecSpinhat(t\!+\!1)\vert\vecSpinhat(t)),\label{eq:PathProb}
\end{eqnarray}
where $\Prob(\vecSpin(t\!+\!1)\vert\vecSpin(t))\!=\!\prod_{i\!=\!1}^N\Prob_{\alpha_i}(\Spin_{i}(t\!+\!1)\vert \Spin_{i_1}(t),..,\Spin_{i_k}(t))$.

The sources of quenched disorder in our model are random Boolean functions and random connections. Boolean functions $\{\alpha_i\}$ are sampled randomly and independently from the distribution
\begin{eqnarray}
\Prob(\alpha)=\sum_{\gamma\in B}p_{\gamma}\;\delta_{\gamma;\alpha},\label{def:gate-disorder}
\end{eqnarray}
where $\sum_{\gamma\in B}p_{\gamma}=1$, $p_{\gamma}\geq0$ and $B$ is the set of all $k$-ary Boolean functions. The connectivity disorder arises from the random sampling of connections generated by selecting the $i$-th function and sampling exactly $k$ indices, $\setI\equiv\{i_1,..,i_k\}$, uniformly from the set of all possible indices $\setN=\{1,\ldots,N\}$. This gives rise to the probabilities
	\begin{eqnarray}
\Prob(\{A_{\setI}^{i}\})\!=\!\frac{1}{Z_A}\prod_{i=1}^{N}\left\{\delta\left[1;\!\!\sum_{\setI^\prime\subseteq\setN} A_{\setI^\prime}^{i}\right]\!\!\prod_{\setI\subseteq\setN}\left\{\frac{1}{N^k}\delta_{A_{\setI}^{i};1}+(1\!-\!\frac{1}{N^k})\delta_{A_{\setI}^{i};0}\right\}\right\},\label{def:connect-disorder}
\end{eqnarray}
where $Z_A$ is a normalization constant. The connectivity tensors $\{A_{i_1,\ldots,i_k}^{i}\}$ define the random topology via entering into the definition of probability~(\ref{eq:micro}) with $\alpha_i(S_{i_1}(t),..,S_{i_k}(t))$ being replaced by $\sum_{i_1,\ldots,i_k}^N A_{i_1,\ldots,i_k}^{i}\alpha_i(S_{i_1}(t),..,S_{i_k}(t))$.  Other connectivity and function profiles can be easily accommodated within our framework by incorporating additional constraints into the definitions (\ref{def:gate-disorder}) and (\ref{def:connect-disorder}) via the appropriate delta functions.
\begin{figure}
\begin{center}
\hspace*{-0mm} \setlength{\unitlength}{0.25mm}
\includegraphics[width=280\unitlength]{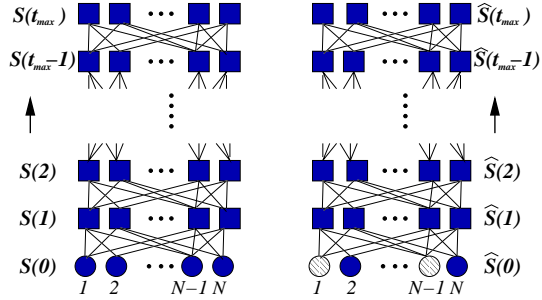}
\caption{The model of two systems with identical topology and different initial conditions. Functions are indicated by squares and input nodes by circles. White (striped) indicates flipped inputs. \label{fig:1}}
\end{center}
\end{figure}
\section{Generating functional analysis\label{section:GFA}}
To analyze the typical properties of the system governed by (\ref{def:algorithm}) we will use the generating functional method of De Dominicis~\cite{dD}. Following the prescription of~\cite{dD} we first define the generating function
\begin{eqnarray}
\Gamma[\vecPsi;\vecPsihat]&=&\left\langle\rme^{-\rmi\sum_{t,i}\{\psi_i(t) S_{i}(t)+\Psihat_i(t) \hat{S}_{i}(t)\}}\right\rangle~,\label{eq:GF}
\end{eqnarray}	
where $\langle\ldots\rangle$ denotes the average over all paths occurring in two systems governed by the joint probability (\ref{eq:PathProb}). The generating function (\ref{eq:GF}) allows us to compute moments of (\ref{eq:PathProb}) by taking partial derivatives with respect to the generating fields $\{\psi_i(t),\Psihat_j(s)\}$, e.g. $\langle S_i(t)\hat{S}_j(s)\rangle=-\lim_{\vecPsi,\vecPsihat\rightarrow\nullvec} \frac{\partial^2}{\partial_{\psi_i(t)}\partial_{\hat{\psi}_j(s)}}\Gamma[\vecPsi;\vecPsihat]$. Secondly, we assume that the system becomes self-averaging, i.e. all thermodynamic macroscopic properties are self-averaging, for $N\rightarrow\infty$~\cite{dD} and  compute $\overline{\Gamma[\vecPsi;\vecPsihat]}$, where $\overline{\cdots}$ is the disorder average; this gives rise to the macroscopic observables
\begin{eqnarray}
&&m(t)\!=\!\frac{1}{N}\sum_{i=1}^N\overline{\langle S_i(t)\rangle}
\hspace{10mm}=\lim_{\vecPsi,\vecPsihat\rightarrow\nullvec}\frac{\rmi}{N}\sum_{i=1}^N \frac{\partial\overline{\Gamma[\vecPsi;\vecPsihat]}}{\partial_{\psi_i(t)}}\label{def:observ}\\
&&C(t,s)\!=\!\frac{1}{N}\!\sum_{i=1}^N\overline{\langle S_i(t)S_i(s)\rangle}=\lim_{\vecPsi,\vecPsihat\rightarrow\nullvec}\!\frac{\!-\!1}{N}\!\sum_{i=1}^N\!\frac{\partial^2\overline{\Gamma[\vecPsi;\vecPsihat]}}{\partial_{\psi_i(t)}\partial_{\psi_i(s)}}\nonumber\\
&&C_{12}(t)\!=\!\frac{1}{N}\!\sum_{i=1}^N\overline{\langle S_i(t)\hat{S}_i(t)\rangle}=\lim_{\vecPsi,\vecPsihat\rightarrow\nullvec}\!\frac{\!-\!1}{N}\!\sum_{i=1}^N\frac{\partial^2\overline{\Gamma[\vecPsi;\vecPsihat]}}{\partial_{\psi_i(t)}\partial_{\hat{\psi}_i(t)}}\nonumber
\end{eqnarray}
where $m(t)$ is the network activity (or magnetization~\cite{MinRBN}), $C(t,\! s)$ is the correlation between two states of the same network  and $C_{12}(t)$  is the overlap between two copies of the same network which is related to the Hamming distance $d(t)$ via $d(t)=\frac{1}{2}(1\!-\!C_{12}(t))$.

Averaging the generating function (\ref{eq:GF}) over the disorder\footnote{Here, in order to present our derivations in a more compact form, we use the superscripts $1$ and $2$ to label two copies of the same system with different noise levels or initial conditions.}
 (see Appendix \ref{section:disorderAverage} for details)  leads us to the saddle-point integral problem
\begin{eqnarray}
\overline{\Gamma}&=&\int\{\mathrm d P \mathrm d\hat{P}\mathrm d\Omega \mathrm d\hat{\Omega}\}\rme^{N\Psi[\{P,\hat{P},\Omega,\hat{\Omega}\}]}\label{eq:GF-recurr-sp}
\end{eqnarray}
where $\Psi$ is the macroscopic saddle-point surface
\begin{eqnarray}
\Psi&=&\rmi\sum_{\{\vecSpin (t)\}}\hat{P}(\{\vecSpin (t)\})P(\{\vecSpin (t)\})\label{def:saddle-recurr}\\
&&+\rmi\int\{\rmd \mathbf{\hat{h}}(t)\}\;\rmd \omega\;\hat{\Omega}(\{\mathbf{\hat{h}}(t)\},\omega)\;\Omega(\{\mathbf{\hat{h}}(t)\},\omega) \nonumber\\
&&+\int\{\mathrm d\mathbf{\hat{h}}(t)\}\;\mathrm d\omega\;\Omega(\{\mathbf{\hat{h}}(t)\},\omega)\sum_{\{\mathbf{S}_j(t)\}}\left\{\prod_{j=1}^k P(\{\vecSpin_j(t)\})\right\}
\nonumber\\
&&\times\overline{\left(\rme^{-\rmi\sum_{t=0}^{t_{max}\!-\!1}\!
\sum_{\gamma =1}^{2}  \hat{h}^{\gamma }(t)\;\alpha(S_{1}^{\gamma}(t),\ldots,S_{k}^{\gamma}(t))-\rmi\omega} -1\right)}^{\;\alpha}-\frac{1}{N}\log Z_A\nonumber\\
&&+\log\sum_{\{\vecSpin (t)\}}\int\{\mathrm d \mathbf{h}(t)\mathrm d \mathbf{\hat{h}}(t)\}\int_{-\pi}^{\pi}\frac{\mathrm d \omega}{2\pi}M[\{\vecSpin (t)\},\{\mathbf{h}(t)\}\vert\{\mathbf{\hat{h}}(t)\}, \omega,\{0\}]
\nonumber
\end{eqnarray}
using the notation $\vecSpin(t)=(S^{1}(t),S^{2}(t))$, $\mathbf{h}(t)=(h^1(t),h^2(t))$, $\mathbf{\hat{h}}(t)=(\hat{h}^1(t),\hat{h}^2(t))$; $M$ is an effective single-site measure
\begin{eqnarray}
&&M[\{\vecSpin(t)\},\{\mathbf{h}(t)\}\vert\{\mathbf{\hat{h}}(t)\}, \omega,\{0\}]
=\Prob(S^1(0),S^2(0))\;\rme^{-\rmi \hat{\Omega}(\{\mathbf{\hat{h}}(t)\},\omega)+\rmi\omega-\rmi\hat{P}(\{\vecSpin(t)\})}\nonumber\\
&&\hspace{57mm}\times\prod_{t=0}^{t_{max}\!-\!1}\!
\prod_{\gamma =1}^{2}\left\{\;\rme^{\rmi  \hat{h}^{\gamma}(t) h^{\gamma}(t)}\frac{\rme^{\beta S^\gamma(t\!+\!1)h_{i}^{\gamma}(t)}}{2\cosh[\beta  h_{i}^{\gamma}(t)]}\right\}\label{def:M}
\end{eqnarray}
with $\Prob(S^1(0),S^2(0))= \frac{1}{4}\left (1+S^1(0)m(0)+S^2(0)\hat m(0)+S^1(0)S^2(0)C_{12}(0)\right)$. The generating fields $\vecPsi, \vecPsihat$ have been removed from the above as they are not needed in the remainder of this calculation. In the limit of $N\rightarrow\infty$ the integral (\ref{eq:GF-recurr-sp}) is dominated by the extremum points of the functional $\Psi$. The functional variation of $\Psi$ with respect to the order parameters $\{P,\hat P, \Omega, \hat\Omega\}$ leads us to the saddle-point equations
 \begin{eqnarray}
&&P(\{\vecSpin(t)\})=\left\langle\prod_{t=0}^{t_{max}\!-\!1}\delta[\vecSpin (t); \vecSpin^\prime (t)]\right\rangle_{M}\label{eq:SP1-recurr}\\
&&\hat{P}(\{\vecSpin (t)\})=\rmi\!\sum_{i=1}^k\sum_{\{\mathbf{S}_j(t)\}}\delta[\{\vecSpin (t)\};\{\vecSpin_i(t)\}]\left\{\prod_{j\neq i}^k P(\{\vecSpin_j(t)\})\right\}\label{eq:SP2-recurr}\\
&&\hspace{17mm}\times\int\{\mathrm d\mathbf{\hat{h}}(t)\}\;\mathrm d\omega\;\Omega(\{\mathbf{\hat{h}}(t)\},\omega) \;\overline{\rme^{-\rmi\sum_{t=0}^{t_{max}\!-\!1}\!
\sum_{\gamma =1}^{2}  \hat{h}^{\gamma }(t)\;\alpha(S_{1}^{\gamma}(t),\ldots,S_{k}^{\gamma}(t))-\rmi\omega}}^{\;\alpha}    \nonumber\\
&&\Omega(\{\mathbf{\hat{h}}(t)\},\omega)=\left\langle\left[\prod_{t=0}^{t_{max}\!-\!1}\delta(\mathbf{\hat{h}}(t)-\mathbf{\hat{h}}^\prime(t))\right]\delta(\omega-\omega^\prime)\right\rangle_{M}\label{eq:SP3-recurr}\\
&&\hat{\Omega}(\{\mathbf{\hat{h}}(t)\},\omega)\!=\!\rmi\!\!\!\!\sum_{\{\mathbf{S}_j(t)\}}\prod_{j=1}^k P(\{\vecSpin_j(t)\})\;\overline{\rme^{-\rmi\sum_{t=0}^{t_{max}\!-\!1}\!
\sum_{\gamma =1}^{2}  \hat{h}^{\gamma }(t)\;\alpha(S_{1}^{\gamma}(t),\ldots,S_{k}^{\gamma}(t))-\rmi\omega}}^{\;\alpha},\label{eq:SP4-recurr}
\end{eqnarray}
where $\left\langle \ldots\right\rangle_M$ is average generated from the single-site measure~(\ref{def:M}). In Appendix~\ref{section:solOfSP} we show that the conjugate order parameter $\hat P$ is a constant. Using this result the saddle-point equation~(\ref{eq:SP4-recurr}) in the single-site measure~(\ref{def:M}) leads us to the main equation of this paper
\begin{eqnarray}
&&\Prob( \vecSpin,\vecSpinhat)\!=\!\Prob(S(0),\hat{S}(0))\!\!\!\!\sum_{\{\mathbf{S}_j,\mathbf{\hat{S}}_j\}}\!\prod_{j=1}^{k}\!\left[\Prob(\vecSpin_j,\vecSpinhat_j)\right]\nonumber\\
&&\times \overline{\prod_{t=0}^{t_{max}-1}\!\!\!\!\Prob_{\alpha}(S(t\!+\!1)\vert S_{1}(t),..,S_{k}(t))\;\Prob_{\alpha}(\hat{S}(t\!+\!1)\vert \hat{S}_{1}(t),..,\hat{S}_{k}(t))}^{\;\alpha}.\label{eq:M}
\end{eqnarray}
The physical meaning of~(\ref{eq:M}) is revealed by $\Prob(\vecSpin,\vecSpinhat)\!=\!\lim_{N\rightarrow\infty}\frac{1}{N}\sum_{i=1}^N\overline{\langle\delta[\vecSpin;\vecSpin_i]\,\delta[\vecSpinhat;\vecSpinhat_i]\rangle}$, i.e. the disorder-averaged joint probability of single-spin trajectories $\vecSpin$ and $\vecSpinhat$ in the two systems. Equation~(\ref{eq:M}) can be used to compute the macroscopic observables (\ref{def:observ}). To demonstrate how this can be done we derive explicitly the expression for the two time correlation $C(t,s)$
\begin{eqnarray}
\label{eq:C12example}
&&\sum_{\mathbf{S},\mathbf{\hat{S}}}\Prob( \vecSpin,\vecSpinhat)\;S(t^\prime\!+\!1)\;S(t^{\prime\prime}\!+\!1)\\
&&=\!\sum_{\mathbf{S},\mathbf{\hat{S}}}\Prob(S(0),\hat{S}(0))\sum_{\{\mathbf{S}_j,\mathbf{\hat{S}}_j\}}\!\prod_{j=1}^{k}\!\left[\Prob(\vecSpin_j,\vecSpinhat_j)\right]S(t^\prime\!+\!1)\;S(t^{\prime\prime}\!+\!1)\nonumber\\
&&\times \overline{\prod_{t=0}^{t_{max}-1}\Prob_{\alpha}(S(t\!+\!1)\vert S_{1}(t),..,S_{k}(t))\;\Prob_{\alpha}(\hat{S}(t\!+\!1)\vert \hat{S}_{1}(t),..,\hat{S}_{k}(t))}^{\;\alpha}\nonumber\\
&&=\!\!\sum_{\{S_{j}(t^\prime),\;{S}_{j}(t^{\prime\prime}\}}\!\prod_{j=1}^{k}\!\left[\Prob(S_{j}(t^\prime),{S}_{j}(t^{\prime\prime})\right]\nonumber\\
&&\times \sum_{S(t^\prime\!+\!1),S(t^{\prime\prime}\!+\!1)}\overline{\Prob_{\alpha}(S(t^\prime\!+\!1)\vert S_{1}(t^\prime),..,S_{k}(t^\prime))\;\Prob_{\alpha}({S}(t^{\prime\prime}\!+\!1)\vert {S}_{1}(t^{\prime\prime}),..,{S}_{k}(t^{\prime\prime}))}^{\;\alpha}\nonumber\\
&&\times \;S(t^\prime\!+\!1)\;S(t^{\prime\prime}\!+\!1)=C(t^\prime\!+\!1,t^{\prime\prime}\!+\!1)\nonumber~;
\end{eqnarray}
note that many of the variables in the summation over $\Prob(\vecSpin,\vecSpinhat)$ are redundant, they have been introduced for methodological reasons but vanish during the derivation. Carrying out a similar derivation for the other order parameters one obtains a closed set of iterative equations:
\begin{eqnarray}
&&m(t\!+\!1)=f_\alpha(m(t))\!=\!\tanh(\beta)\sum_{S}\prod_{j=1}^{k}\left[\frac{1\!+\! S_jm(t)}{2}\right] \overline{\alpha(S)}^{\;\alpha}\label{eq:m}\\
&&C(t\!+\!1,s\!+\!1)\!=\!F_\alpha(m(t),m(s),C(t,s))\nonumber\\
&&\!=\!\tanh^2(\beta)\sum_{S,\hat{S}}\!\prod_{j=1}^{k}\!\left[\!\frac{1\!+\! S_j m(t)\!+\!\hat{S}_j m(s) \!+\! \!S_j\hat  S_j C(t,s)}{4}\right] \overline{ \alpha(S)\;\alpha(\hat{S})}^{\;\alpha}\label{eq:Corr}\\
&&C_{12}(t\!+\!1)=F_\alpha(m(t),\hat m(t),C_{12}(t)),\label{eq:overlap}
\end{eqnarray}
where $S=(S_1,\ldots,S_k)$ and the equation for $\hat m(t)$ is the same as (\ref{eq:m}).

\section{Results\label{section:results}}
In this section, we first apply the equations~(\ref{eq:m})-(\ref{eq:overlap}) to the recurrent Boolean networks with thermal noise. We recover results of the annealed approximation for the order parameters $m$ and $C_{12}$. However, the two-time correlation function $C(t,s)$, computed here for the first time, allows us to study properties of the stationary states.  Furthermore, the exactness of our method allows us to derive a rigorous upper bound on the noise level above which the system is always ergodic. In addition, we use the equation (\ref{eq:M}) to study models with strong memory effects where the annealed approximation method is no longer valid.

\subsection{Boolean networks}

\subsubsection{Stationary states}

It is clear from the results (\ref{eq:m})-(\ref{eq:overlap}) and (\ref{eq:M}) that the evolution of all many-time single-site correlation functions is driven by the magnetization $m(t)$. A similar scenario was observed in recurrent asymmetric neural networks~\cite{MimuraAndCoolen} which have a similar topology but uses different update functions than model (\ref{def:algorithm}) . The asymmetric neural network model can be seen as a special case of the N-k model when only \emph{linear threshold  Boolean functions} are used and the thermal noise enters into the system via randomness in the thresholds. Furthermore, for the stationary  solution $m\!=\!f_\alpha(m)$ ($m\!=\!\lim_{t\rightarrow\infty} m(t)$) the solution of $q\!=\!F_\alpha(m,m,q)$ [here $q\!=\!\lim_{t\rightarrow\infty}\lim_{t_w\rightarrow\infty}C(t\!+\!t_w,t_w)$ is the Edwards-Anderson order parameter, used in disordered systems~\cite{BOOK} to detect the spin glass (SG) phase where $m\!=\!0$ and $q\!\neq\!0$] and $C_{12}\!=\!F_\alpha(m,m,C_{12})$ are identical. This was also observed in asymmetric neural networks~\cite{adnn} and because of the equality $q\!=\!C_{12}$ there is only one average distance $\frac{1}{2}(1\!-\!q)$ on the attractor~\cite{adnn} and all points in the basin of attraction uniformly cover the stationary states (see Figure \ref{fig:2}).
\begin{figure}
\begin{center}
\setlength{\unitlength}{0.25mm}
\includegraphics[width=280\unitlength]{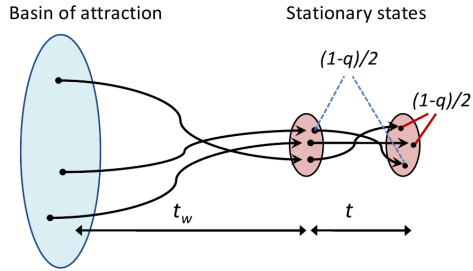}
\caption{Topology of the basin of attraction with its stationary states. \label{fig:2}}
\end{center}
\end{figure}

\subsubsection{Annealed model\label{section:AnnModel}}
The annealed approximation method, where connectivities and Boolean functions change at each time step of the process (\ref{def:algorithm}), provides \emph{identical results} for $m$ and $C_{12}$ to those of (\ref{eq:m}) and (\ref{eq:overlap})~\cite{Kesseli}. However, within annealed approximation the two-time correlations take the form $C(t,s)\!=\!m(t)m(s)$, where $t\!>\!s$, which is the solution of (\ref{eq:Corr}) \emph{only when networks are constructed from a single Boolean function.} This result follows from the equality $C(t,0)\!=\!m(t)m(0)$, which is clear from the equations (\ref{eq:M}) and  (\ref{eq:C12example}), and the fact that for a single function, in the absence of an average over $\alpha$, the joint probability of two spins in the equation (\ref{eq:Corr}) factorizes when $C(t,s)\!=\!m(t)m(s)$.

The classical annealed approximation result~\cite{Derrida} for RBN, which is exact in this case~\cite{DerridaAndW,Hilhorst}, can be easily recovered from  the equations (\ref{eq:m})-(\ref{eq:overlap}) using the property $ \overline{\alpha(S)}^{\;\alpha}\!=\!0$ for all $S\!\in\!\{-1,1\}^k$ and $\overline{\alpha(S)\alpha(\hat S)}^{\;\alpha}\!=\!0$, for all $S\!\neq\!\hat S$ where the $\alpha$ average is taken over all Boolean functions with equal weight.  In the noisy case ($\beta\!<\!\infty$), the magnetization variable $m(t)\!=\!0$ for all $t\!>\!0$ and $q=\tanh^2(\beta)(\frac{1\!+\!q}{2})^k$, corresponding to the stationary solution of (\ref{eq:Corr}), has one stable solution $q\!\neq\!0$ for all finite $\beta>0$ and $k$. For $\beta\!\rightarrow\!\infty$ (no noise), there is a transition from one stable solution $q=1$ for $k\!\leq\!2$ to two solutions $q\!=\!1$ (unstable) and $q\!\neq\!0$ (stable) for $k\!>\!2$~\cite{Derrida}.

\subsubsection{Noise upper bound}

An interesting question related to the ergodicity and phase transitions is whether the system~(\ref{def:algorithm}) can retain information about its initial state in the presence of noise. This question has received a considerable attention in the works on cellular automata~\cite{Pippenger} and in a closely related field of fault-tolerant computation~\cite{Evans:MTNK,NoisyPRL}.

The unordered paramagnetic (PM) phase $m\!=\!0$, where no information is retained, is a fixed point of~(\ref{eq:m}) only when  $\sum_S \overline{\alpha(S)}^{\;\alpha}\!=\!0$.
\begin{proposition}\label{prop:1}
The point $m\!=\!0$ is a stable and unique solution of (\ref{eq:m}) when $\tanh\beta\!<\!\left\{2^{k\!-\!1}/k\binom{k\!-\!1}{(k\!-\!1)/2}; 2^{k\!-\!2}/(k\!-\!1)\binom{k\!-\!2}{(k\!-\!2)/2}\right\}\!\equiv \!b(k)$ for $k$ odd and even respectively.
\end{proposition}
\begin{proof}
To prove this we first find a Boolean function $\chi$ such that $f_\chi(m)\!\geq\! f_\alpha(m)$ when $m\!\in\![0,1)$ and $f_\alpha(m)\!\geq\! f_\chi(m)$ when $m\!\in\!(-1,0]$. It turns out that any function from the set $\chi(S)\!=\!\sgn[\sum_{j=1}^k S_j]\!+\!\delta[0;\sum_{j=1}^k S_j]\gamma(S)$, where $\gamma(S)\in\{\!-\!1,1\}$ and such that $\sum_S\delta[0;\sum_{j=1}^k S_j]\gamma(S)\!=\!0$\footnote{We use the convention $\sgn[0]\!=\!0$ throughout this paper.} satisfies these properties. To show this we define the average $\left\langle \cdots\right\rangle_{S\vert m}=\sum_{S}\prod_{j=1}^{k}\left[\frac{1\!+\! S_jm}{2}\right](\cdots)$ and use the shorthand notations $\left\{\1_+[S],\1_-[S],\1_0[S]\right\}$ for the indicator functions $\{\1[\sum_{j=1}^k S_j\!>\!0],\1[\sum_{j=1}^k S_j\!<\!0],\1[\sum_{j=1}^k S_j\!=\!0]\}$.  Then we compute the difference $\Delta(m)=(f_\chi(m)-f_\alpha(m))/4\tanh(\beta)$ as follows
\begin{eqnarray}
\Delta(m)&=&\frac{1}{4}\left\langle \sgn\left[\sum_{j=1}^k S_j\right]-\overline{\alpha(S)}^{\;\alpha}\right\rangle_{S\vert m}\label{eq:difference}\\
&=&\frac{1}{4}\left\langle \1_+\left[S\right]-\1_-\left[S\right]
-\left(\1_+\left[S\right]+\1_-\left[S\right]+\1_0\left[S\right]\right)\overline{\alpha(S)}^{\;\alpha}\right\rangle_{S\vert m}\nonumber\\
&=&\frac{1}{2}\left\langle \1_+\left[S\right]\frac{1}{2}(1\!-\!\overline{\alpha(S)}^{\;\alpha})-\1_-\left[ S\right]\frac{1}{2}(1\!+\!\overline{\alpha(S)}^{\;\alpha})-\frac{1}{2}\1_0\left[S\right]\overline{\alpha(S)}^{\;\alpha}\right\rangle_{S\vert m}\nonumber\\
&=&\frac{1}{2}\left\langle
\1_+\left[S\right]\overline{\1\left[\alpha(S)\!=\!-\!1\right]}^{\;\alpha}
\!-\!\1_-\left[S\right]\overline{\1\left[\alpha(S)\!=\!+\!1\right]}^{\;\alpha}\!-\!\frac{1}{2}\1_0\left[S\right] \overline{\alpha(S)}^{\;\alpha}\right\rangle_{S\vert m}\nonumber\\
&=&\frac{1}{2}\left(\left[\frac{1\!+\!m}{2}\right]\left[\frac{1\!-\!m}{2}\right]\right)^{\frac{k}{2}}\Bigg\{\sum_{S}\left[\frac{1\!+\!m}{1\!-\!m}\right]^{\frac{\vert\sum_{j=1}^k S_j\vert}{2}}\1_+\left[S\right]\overline{\1\left[\alpha(S)\!=\!-\!1\right]}^{\;\alpha}\nonumber\\
&\!-\!&\sum_{S}\left[\frac{1\!-\!m}{1\!+\!m}\right]^{\frac{\vert\sum_{j=1}^k S_j\vert}{2}}
\1_-\left[S\right]\overline{\1\left[\alpha(S)\!=\!+\!1\right]}^{\;\alpha}-\frac{1}{2}\sum_{S}\1_0\left[S\right]  \overline{\alpha(S)}^{\;\alpha}\Bigg\},\nonumber
\end{eqnarray}
where in the above we used the equality $\prod_{j=1}^{k}\left[\frac{1\!+\! S_jm}{2}\right]=\left[\frac{1\!+\!m}{2}\right]^{\frac{k+\sum_{j=1}^k S_j}{2}}\left[\frac{1\!-\!m}{2}\right]^{\frac{k-\sum_{j=1}^k S_j}{2}}$. Let us now consider the sum
\begin{eqnarray}
&&\sum_{S}\overline{\alpha(S)}^{\;\alpha}=\sum_{S}\overline{\1\left[\alpha(S)\!=\!+\!1\right]}^{\;\alpha} \!-\! \sum_{S}\overline{\1\left[\alpha(S)\!=\!-\!1\right]}^{\;\alpha}\label{eq:zero}\\
&=&\sum_{S}\left(\1_+\left[S\right]+\1_-\left[S\right]+\1_0\left[S\right]\right)\left(\overline{\1\left[\alpha(S)\!=\!+\!1\right]}^{\;\alpha}-\overline{\1\left[\alpha(S)\!=\!-\!1\right]}^{\;\alpha}\right)\nonumber\\
&=&\sum_{S}\left(\frac{1}{2}\1_0\left[S\right]\overline{\alpha(S)}^{\;\alpha}+\1_-\left[S\right]\overline{\1\left[\alpha(S)\!=\!+\!1\right]}^{\;\alpha}\!-\!\1_+\left[S\right]\overline{\1\left[\alpha(S)\!=\!-\!1\right]}^{\;\alpha}\right)=0.\nonumber
\end{eqnarray}

Adding the above representation of zero to the terms inside the curly brackets in the equation (\ref{eq:difference}) gives

\begin{eqnarray}
\Delta(m)&=&\frac{1}{2}\left(\left[\frac{1\!+\!m}{2}\right]\left[\frac{1\!-\!m}{2}\right]\right)^{\frac{k}{2}}\nonumber\\
&&\times\Bigg\{\sum_{S}\left(\left[\frac{1\!+\!m}{1\!-\!m}\right]^{\frac{\vert\sum_{j=1}^k S_j\vert}{2}}-1\right)\1_+\left[S\right]\overline{\1\left[\alpha(S)\!=\!-\!1\right]}^{\;\alpha}\nonumber\\
&\!+\!&\sum_{S}\left(1-\left[\frac{1\!-\!m}{1\!+\!m}\right]^{\frac{\vert\sum_{j=1}^k S_j\vert}{2}}\right)
\1_-\left[S\right]\overline{\1\left[\alpha(S)\!=\!+\!1\right]}^{\;\alpha}\Bigg\}\label{eq:diffFinal}
\end{eqnarray}
which is clearly $\Delta(m)\geq0$ for $m\!\in\![0,1)$ and $\Delta(m)\leq0$ for $m\!\in\!(-1,0]$.

One can show that the function $f_\chi(m)$, which we used in the bounding procedure (\ref{eq:difference})-(\ref{eq:diffFinal}),  has the following properties:\footnote{This can be done by applying the steps of Lemma 1 in \cite{Evans:MTNK} to the function $f_\chi(m)$.} (i) $m\!>\!f_\chi(m)$ when  $m\!\in\!(0,1)$ and $f_\chi(m)\!>\!m$ when $m\!\in\!(-1,0)$ ($f_\chi(0)\!=\!0$) for $\tanh\beta\!<\!b(k)$; (ii) for $\tanh\beta\!>\!b(k)$ there exists $m^*\in[-1,1]\setminus\{0\}$ such that $f_\chi(m^*)\!=\!m^*$.\qed
\end{proof}

The consequence of Proposition \ref{prop:1} is that the ordered ferromagnetic (FM) phase $m\!\neq\!0$  is a fixed point of (\ref{eq:m}) (if at all) \emph{only} for values of $\beta$ and $k$ which satisfy $\tanh\beta\!>\!b(k)$. This situation leads to the PM/FM phase boundary, depicted in the phase diagram (Figure \ref{fig:3}), which for $k\rightarrow\infty$ approaches $p=1/2$ as $1/2-p(k)=O(1/\sqrt{k})$; this is follows from the Stirling's approximation of $b(k)$. BN constructed from a single Boolean function $\chi(S)$ saturates this boundary. We note that a similar result, for odd $k$ only, have been conjectured using the annealed approximation and multiplexing techniques~\cite{Peixoto}.

\begin{figure}
\begin{center}
\hspace{-30mm}
\setlength{\unitlength}{1.4mm}
\begin{picture}(60,55)
\put(5,0){\epsfysize=45\unitlength\epsfbox{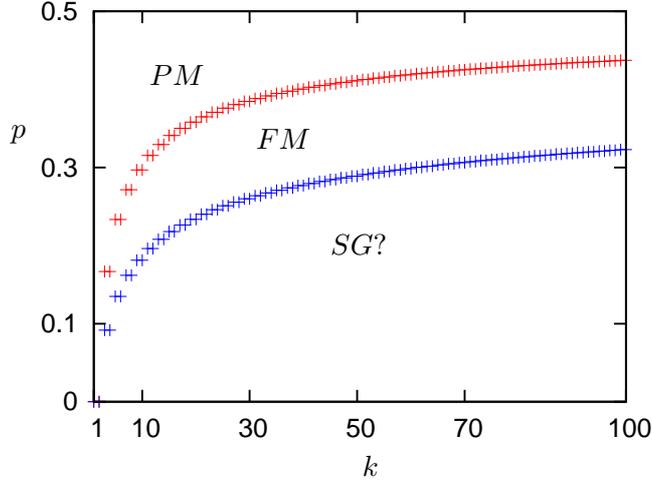}}
\put(7,30){$p$}\put(37,19){$SG?$}\put(30,29){$FM$}\put(20,35){$PM$}
\put(40,-2){$k$}
\end{picture}
\end{center}
\caption{Phase diagram for the recurrent Boolean network governed by (\ref{def:algorithm}). The model is paramagnetic (PM) for any distribution of Boolean functions for the noise parameter $p>(1-b(k))/2$. It can be ferromagnetic (FM) below the boundary $(1-b(k))/2$ and there is a possibility of the spin-glass (SG) phase when  $p<(1-\sqrt{b(k)})/2$.\label{fig:3}}
\end{figure}

In the original RBN model, which we have considered in the section \ref{section:AnnModel}, the stationary state $m\!=\!0$, $q\neq0$ is for any $\beta\in(0,\infty)$. Here we explore a possibility for the system (\ref{def:algorithm}) to have the disordered PM ($m\!=\!0$, $q\!=\!0$) and two ordered FM ($m\!\neq\!0$, $q\!\neq\!0$) and SG ($m\!=\!0$, $q\!\neq\!0$) states. For $\lim_{t\rightarrow\infty}m(t)\!=\!m$, $q\!=\!0$ is a fixed point of (\ref{eq:Corr}) iff $ \overline{\{\sum_S\alpha(S)\}^2}^{\;\alpha}\!=\!0$ which occurs only for \emph{balanced Boolean functions}, with an equal number of $\!\pm\!1$ in the output.
\begin{proposition}\label{prop:2}
 For $m\!=\!0$ the point $q\!=\!0$ is a unique stable solution of (\ref{eq:Corr}) when $\tanh^2\beta\!<\!b(k)$.
\end{proposition}
\begin{proof}
In order to show this we first define the function $T(C)\!=\!\tanh^2(\beta)\sum_{S,\hat{S}}\!\prod_{j=1}^{k}\!\left[\!\frac{1+S_j\hat  S_j C}{4}\right]\sgn[S\cdot\hat{S}]$ which is related to the function $f_\chi$ via the equality $T(C)=\tanh(\beta)f_\chi(C)$. This property follows from the calculation
\begin{eqnarray}
T(C)&=&\tanh^2(\beta)\sum_{S,\hat{S}}\!\prod_{j=1}^{k}\!\left[\!\frac{1+S_j\hat  S_j C}{4}\right]\sgn[S.\hat{S}]\\
&=&\tanh^2(\beta)\sum_{S}\!\prod_{j=1}^{k}\!\left[\!\frac{1+S_jC}{2}\right]\sgn[\sum_{j=1}^k S_j]=\tanh(\beta)f_\chi(C).\nonumber
\end{eqnarray}
Next we define the function $g(C)\!=\!\tanh^2(\beta)\sum_{S,\hat{S}}\!\prod_{j=1}^{k}\!\left[\!\frac{1+S_j\hat  S_j C}{4}\right]\alpha(S)\alpha(\hat{S})$, where $\alpha$ is an arbitrary balanced Boolean function, and compute the difference $\Delta(C)=(T(C)-g(C))/4\tanh^2(\beta)$ using the same steps as described in the equations (\ref{eq:difference})-(\ref{eq:diffFinal}). The result of this computation is that $\Delta(C)\geq0$ and $\Delta(C)\leq0$ on the intervals $C\!\in\![0,1)$ and $C\!\in\!(-1,0]$, respectively, from which the bounds $T(C)\geq F_\alpha(0,0,C)$ and $T(C)\leq F_\alpha(0,0,C)$ on the same intervals follow. The behavior of $T[C]$ with respect to the inverse temperature $\beta$ is the same as of $f_\chi$, which we described in the proof of Proposition \ref{prop:1}, but with the $\tanh(\beta)$ being replaced by the $\tanh^2(\beta)$.\qed

\end{proof}
From the Proposition \ref{prop:2} the case of $m\!=\!0$, $q\!\neq\!0$ and finite  $\beta$ occurs only (if at all) when $\tanh^2\beta\!>\!b(k)$. The resulting FM/SG phase boundary (see Figure~\ref{fig:3}) approaches $p=1/2$ as $1/2-p(k)=O(k^{-1/4})$ when $k\rightarrow\infty$. The $\alpha$-averages in equations (\ref{eq:m})-(\ref{eq:Corr}) can be computed for a uniform distribution over all balanced Boolean functions to obtain $m(t)\!=\!0$ for all $t\!>\!0$, which implies $q\!=\!\tanh^2(\beta)\left((\frac{1\!+\!q}{2})^k(1\!+\!\frac{1}{2^k\!-\!1})\!-\!\frac{1}{2^k\!-\!1}\right)$. The latter has only one $q\!=\!0$ trivial solution for any finite $\beta$ and develops a second $q\!=\!1$ solution only for $\beta\!\rightarrow\!\infty$. Thus, only the model (\ref{def:algorithm}) with non-uniform distributions over the balanced Boolean functions can have the critical behavior as in Figure \ref{fig:3}.

It is interesting that the upper bound $b(k)$ computed here for $k$ odd is identical to the one computed for noisy Boolean formulas~\cite{Evans:MTNK}. A noisy Boolean formula is a tree in which leaves are either Boolean constants or references to arguments, internal nodes are noisy Boolean functions (for each function-input there is an error probability $p$ which inverts the function-output) and the root corresponds to the formula output. The MAJ-$k$ function, which plays a prominent role in the area of fault-tolerant computation (FTC) as it allows to correct the errors by its majority-vote function~\cite{VonNeumann}, saturates the bound $b(k)$. The idea to have two copies of the same system, used in this work only to study initial-states perturbations,  is also useful for FTC as it allows to compare the noisy system against its noiseless counterpart~\cite{NoisyPRL}.

The connection of our work with FTC stems from the fact that each site $i$ at time $t$ in our model can be associated with the output $S_i(t)$ of a $k$-ary Boolean formula of depth $t$  which computes a function of the associated initial states (a subset of $\{S_i(0)\}$)~\cite{DerridaAndW}. In the presence of noise, a formula of considerable depth (large $t$) loses all input information for  $\tanh\beta\!<\!b(k)$ and odd $k$~\cite{Evans:MTNK}. This suggests that the upper bound $b(k)$, for odd  $k$, is more general and is valid for transitions at \emph{all} $m$ values identifying the point where stationary states depend on the initial states and ergodicity breaks.  For $k$ even such \emph{general} threshold is not yet known.

\subsection{Boolean networks with memory}
In model~(\ref{def:algorithm}) the state of site $i$ at time $t$ depends on its states at previous times only indirectly via the sites affected by the state of site $i$ at previous times. These dependencies create correlations via the directed loops in the time-space picture of BN (as in Figure~\ref{fig:1}), but in the limit of  $N\!\rightarrow\!\infty$ they become very weak, as was argued in previous works in this area~\cite{DerridaAndW}. This allows one to express the observables of interest (\ref{def:observ})  in the closed form (\ref{eq:m})-(\ref{eq:overlap}). However, in a broad family of models, which includes the Boolean networks with reversible computation~\cite{Revers} and gene networks with self-regulation~\cite{DeSales}, the state of a site $i$ at a time $t\!+\!1$ depends directly on its state at time $t$.

\subsubsection{Random threshold networks}
An exemplar model with strong memory effects, which was used in~\cite{Li} to construct a model of cell-cycle regulatory network ($N\!=\!11$) of budding yeast, is of the form
\begin{equation}
S_i(t\!+\!1)\!=\!\sgn[h_i(t)\!-\!2h]\!+\!S_i(t)\delta[h_i(t);2h]\label{eq:process},
\end{equation}
where $h_i(t)\!=\!\sum_{j=1}^k\xi_{i_j}(1\!+\!S_{i_j}(t))$ and $\xi_{i_j}\!\in\!\{\!-\!1,1\}$. Mean-field theory ($N\!\rightarrow\!\infty$) was derived~\cite{RTN} using the annealed approximation in a variant of this model, where the interactions $\{\xi_j\}$ were randomly distributed $\Prob(\xi_j\!=\!\pm1)\!=\!1/2$. Significant discrepancies between the theory and simulation results has been pointed out~\cite{RTN} for  integer $h$ values (in this case it is possible that $2h\!=\!h_i(t)$), which was attributed to the presence of strong memory effects. Refinements of the annealed approximation method improved the results obtained only slightly~\cite{Heckel,Zanudo} but break down in most of the parameter space.

This model (\ref{eq:process}) can be easily incorporated into our theoretical framework. The result of the GFA (\ref{eq:M}) for this process (with thermal noise) can be obtained by replacing the  average $\overline{(\cdots)}^\alpha$ by $\overline{(\cdots)}^\xi$ and the probability function $\Prob_{\alpha}(S(t\!+\!1)\vert S_{1}(t),..,S_{k}(t))$ by
\begin{eqnarray}
\Prob_{\xi}(S(t\!+\!1)\vert S(t);S_1(t),..,S_k(t))=\frac{\rme^{\beta S(t\!+\!1)\{\sgn[h(t)\!-\!2h]\!+\!S(t)\delta[h(t);2h]\}}}{2\cosh\beta\{\sgn[h(t)\!-\!2h]\!+\!S(t)\delta[h(t);2h]\}}\label{def:ProbRTN},
\end{eqnarray}
where $h(t)\!=\!\sum_{j=1}^k\xi_{j}(1\!+\!S_{j}(t))$.

In the case of $h\!\in\!\mathbb{R}$, the probability function (\ref{def:ProbRTN}) is independent of $S(t)$  and equations (\ref{eq:m})-(\ref{eq:overlap}) have the same structure as model~(\ref{eq:process}): the $\alpha$-averages $\overline {\alpha(S)}^\alpha$ and $\overline{\alpha(S)\alpha(\hat S)}^\alpha$ are replaced by the averages $\overline{\sgn[h(t)\!-\!2h]}^\xi$ and $\overline{\sgn[h(t)\!-\!2h]\sgn[\hat h(t)\!-\!2h]}^\xi$, respectively. The equation for $m(t)$ recovers the annealed approximation result~\cite{RTN} (using the relation $b(t)\!=\!(1\!+\!m(t))/2$). In Fig.~\ref{fig:4}(a,b), we plot our analytical predictions for the evolution of $m(t)$ and $C(t\!+\!t_w,t_w)$ against the results of Monte Carlo (MC) simulation which use (\ref{eq:process}). The correlation function $C(t\!+\!t_w,t_w)$, in the limit of  $t\!\rightarrow\!\infty,t_w\!\rightarrow\!\infty$, approaches the stationary solution of the overlap function (\ref{eq:overlap}) with increasing $t_w$ as predicted (Fig.~\ref{fig:4}(b)).
\begin{figure}[t]
\setlength{\unitlength}{0.4mm}
\begin{center}
\begin{picture}(350,100)
\put(159,0){\includegraphics[height=100\unitlength,width=160\unitlength]{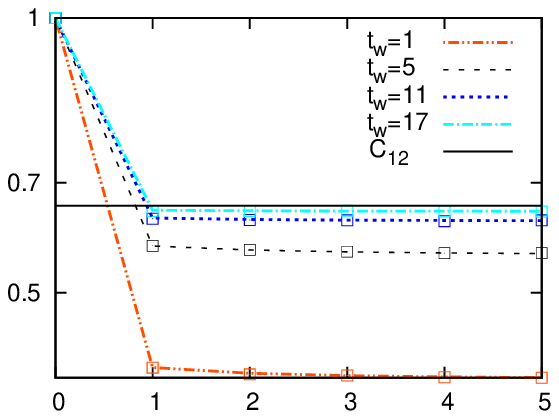}}
\put(245,-7){$t$}
\put(0,0){\includegraphics[height=100\unitlength,width=160\unitlength]{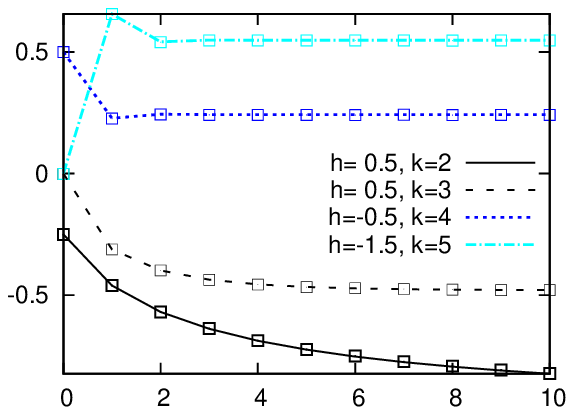}}
 \put(1,65){$m$}
\put(85,-7){$t$}  \put(160,65){$C$}
\put(130,80){$(a)$}\put(290,80){$(b)$}
\end{picture}
\end{center}
\caption{Evolution of the magnetization ($m\equiv m(t)$) and correlation ($C\equiv C(t\!+\!t_w,t_w)$) functions with time $t$ is governed by (\ref{eq:process}). Theoretical results (lines) are plotted against the results of MC simulations (symbols) with $N\!=\!10^5$. Each MC data-point is averaged over 10 runs. Error bars are smaller than symbol size. Evolution of $m$ (a) and $C$ (b) for $h\!\in\!\mathbb{R}$. In (b) we plot $C$ for $h\!=\!0.5$ and $k\!=\!3$. \label{fig:4} 
}
\end{figure}

The situation is very different when $h\!\in\!\mathbb{Z}$. The magnetization $m(t)\!=\!\sum_{\vecSpin}\Prob( \vecSpin)S(t)$, where $\Prob(\vecSpin)$ is a marginal of (\ref{eq:M}) with $P_\alpha\!\rightarrow\! P_\xi$, is no longer closed as in (\ref{eq:m}), but depends on $2^{t\!-\!1}\!-\!1$ macroscopic observables (all magnetization, all multi-time correlations). Thus the number of macroscopic observables that determine the value of  $m(t)$, or any other function computed from (\ref{eq:M}), grows exponentially with time. Annealed approximation results~\cite{RTN} for this model when $h\!\in\!\mathbb{Z}$ are only exact up to $t\!<\!2$ time steps (the equation for $b(1)\!=\!(1\!+\!m(1))/2$ in our approach and in~\cite{RTN} are identical) and deviate significantly from the exact solution at later times (Fig.~\ref{fig:5}(c)). A typical evolution of the correlation function $C(t\!+\!t_w,t_w)$ in the system (\ref{eq:process}) when $h\in\mathbb{Z}$ is shown in Fig.~\ref{fig:5}(d).
\begin{figure}
\setlength{\unitlength}{0.4mm}
\begin{center}
\begin{picture}(350,100)
\put(159,0){\includegraphics[height=100\unitlength,width=160\unitlength]{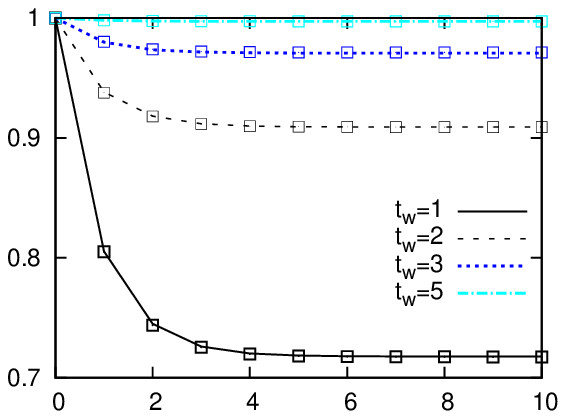}}
\put(245,-7){$t$}
\put(0,0){\includegraphics[height=100\unitlength,width=160\unitlength]{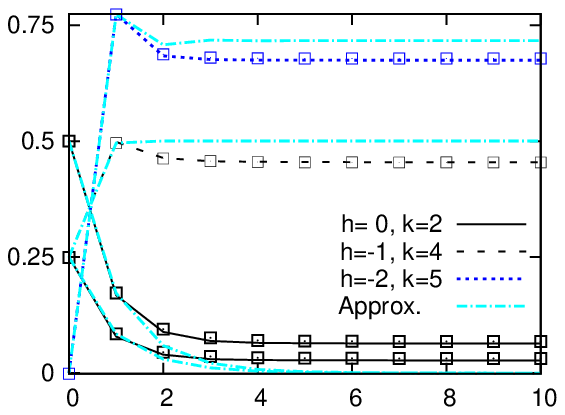}}
 \put(1,65){$m$}
\put(85,-7){$t$}  \put(160,65){$C$}
\put(130,80){$(a)$}\put(290,80){$(b)$}
\end{picture}
\end{center}
\caption{Evolution of $m\equiv m(t)$ (a) and $C\equiv C(t\!+\!t_w,t_w)$ (b) for $h\!\in\!\mathbb{Z}$. In (b) we plot $C$ for $h\!=\!0$ and $k\!=\!2$. Theoretical results (lines) are plotted against the results of MC simulations (symbols) with $N\!=\!10^5$. Each MC data-point is averaged over 10 runs. Error bars are smaller than symbol size.\label{fig:5} 
}
\end{figure}

\section{Discussion}
\label{section:summary}
We applied the generating functional method to analyze the dynamics of recurrent Boolean networks. The analysis resulted in a coupled set of recursive equations for a small number of macroscopic observables that provide an exact description of the dynamics in a broad range of Boolean networks. Based on the analysis we
also showed that for a large class of models the annealed approximation does provide exact results for both magnetization and overlap order parameters; although results for correlation between states at different times are generally incorrect. However, it turns out that in models where the state of spin at time $t$ depends on its state at previous times directly the annealed approximation always fails. This is due to the fact that approximation does not take into account the correlations which are very strong in such models. Comparing the two transition probabilities (\ref{eq:micro}) and (\ref{def:ProbRTN}) for the processes without and with memory, respectively, it is clear from the single-spin trajectory equation (\ref{eq:M}), that as soon as there is an explicit dependence of a spin on its states at previous times, an exponential (in time) number of macroscopic observables will be required.  Furthermore, also in models where this approximation works well it is useful to know the two-time correlations as it provides us with insight into properties of the stationary states.

When one considers systems with thermal noise, the suggested framework provides additional new and interesting  results. We have computed the noise threshold above which the system is always ergodic and   critical noise levels where phase transitions occur. Here the two-time correlation function is particularly useful as it allows one to compute the Edwards-Anderson order parameter $q$, used to detect the spin-glass phase.

One of the remaining questions is to provide an example of a model (or show that it does not exists)  with a phase diagram as in Figure~\ref{fig:3}. The direct computation of $q$ for all balanced Boolean functions with $k\geq3$ inputs is possible for small $k$ but soon becomes intractable as the number of such functions grows exponentially with $k$. The other important question is to find a systematic way to generate good approximations for the dynamics with strong memory effects. Although the theory developed in this work can describe the dynamics of such models exactly it can be used only for relatively short times due to the rapid increase in the number of macroscopic order parameters. However, the equations of our theory  can be used to check the quality of approximations used in such models in all future works and possibly serve as a starting point for such studies. The work undertaken here can be extended in a numerous ways. For instance, one can easily adapt the framework developed here to study Boolean networks with inhomogeneous connectivities~\cite{scaleFreeBN} and to examine different noise models~\cite{NoisyPRE}.

\section*{Acknowledgements}
Support by the Leverhulme trust (grant F/00 250/H) is acknowledged.


\begin{thebibliography}{34}
\expandafter\ifx\csname natexlab\endcsname\relax\def\natexlab#1{#1}\fi
\expandafter\ifx\csname bibnamefont\endcsname\relax
  \def\bibnamefont#1{#1}\fi
\expandafter\ifx\csname bibfnamefont\endcsname\relax
  \def\bibfnamefont#1{#1}\fi
\expandafter\ifx\csname citenamefont\endcsname\relax
  \def\citenamefont#1{#1}\fi
\expandafter\ifx\csname url\endcsname\relax
  \def\url#1{\texttt{#1}}\fi
\expandafter\ifx\csname urlprefix\endcsname\relax\def\urlprefix{URL }\fi
\providecommand{\bibinfo}[2]{#2}
\providecommand{\eprint}[2][]{\url{#2}}

\bibitem[{\citenamefont{{Kauffman}}(1969{\natexlab{a}})}]{Kauffman}
\bibinfo{author}{\bibfnamefont{S.~A.} \bibnamefont{{Kauffman}}},
  \bibinfo{journal}{J. Theor. Biol.} \textbf{\bibinfo{volume}{22}},
  \bibinfo{pages}{437} (\bibinfo{year}{1969}{\natexlab{a}}).

\bibitem[{\citenamefont{Kauffman}(1993)}]{RBNbook}
\bibinfo{author}{\bibfnamefont{S.}~\bibnamefont{Kauffman}},
  \emph{\bibinfo{title}{The Origins of Order}} (\bibinfo{publisher}{Oxford
  University Press}, \bibinfo{address}{New York}, \bibinfo{year}{1993}).

\bibitem[{\citenamefont{{Moreira} et~al.}(2004)\citenamefont{{Moreira},
  {Mathur}, {Diermeier}, and {Amaral}}}]{Moreira}
\bibinfo{author}{\bibfnamefont{A.~A.} \bibnamefont{{Moreira}}},
  \bibinfo{author}{\bibfnamefont{A.}~\bibnamefont{{Mathur}}},
  \bibinfo{author}{\bibfnamefont{D.}~\bibnamefont{{Diermeier}}},
  \bibnamefont{and} \bibinfo{author}{\bibfnamefont{L.~A.~N.}
  \bibnamefont{{Amaral}}}, \bibinfo{journal}{Proc. Nat. Acad. Sci. U.S.A.}
  \textbf{\bibinfo{volume}{101}}, \bibinfo{pages}{12085}
  (\bibinfo{year}{2004}).

\bibitem[{\citenamefont{{Kauffman}}(1969{\natexlab{b}})}]{Nature}
\bibinfo{author}{\bibfnamefont{S.}~\bibnamefont{{Kauffman}}},
  \bibinfo{journal}{Nature} \textbf{\bibinfo{volume}{224}},
  \bibinfo{pages}{177} (\bibinfo{year}{1969}{\natexlab{b}}).

\bibitem[{\citenamefont{{Derrida} et~al.}(1987)\citenamefont{{Derrida},
  {Gardner}, and {Zippelius}}}]{DerridaANN}
\bibinfo{author}{\bibfnamefont{B.}~\bibnamefont{{Derrida}}},
  \bibinfo{author}{\bibfnamefont{E.}~\bibnamefont{{Gardner}}},
  \bibnamefont{and}
  \bibinfo{author}{\bibfnamefont{A.}~\bibnamefont{{Zippelius}}},
  \bibinfo{journal}{Europhys. Lett.} \textbf{\bibinfo{volume}{4}},
  \bibinfo{pages}{167} (\bibinfo{year}{1987}).

\bibitem[{\citenamefont{Aldana et~al.}(2003)\citenamefont{Aldana, Coppersmith,
  and Kadanoff}}]{Aldana}
\bibinfo{author}{\bibfnamefont{M.}~\bibnamefont{Aldana}},
  \bibinfo{author}{\bibfnamefont{S.}~\bibnamefont{Coppersmith}},
  \bibnamefont{and} \bibinfo{author}{\bibfnamefont{L.~P.}
  \bibnamefont{Kadanoff}}, \emph{\bibinfo{title}{Perspectives and Problems in
  Nonlinear Science. A Celebratory Volume in Honor of Lawrence Sirovich.}}
  (\bibinfo{publisher}{Springer}, \bibinfo{address}{New York},
  \bibinfo{year}{2003}), chap. \bibinfo{chapter}{Boolean dynamics with random
  couplings}, pp. \bibinfo{pages}{23--89}.

\bibitem[{\citenamefont{{Drossel}}(2008)}]{Drossel}
\bibinfo{author}{\bibfnamefont{B.}~\bibnamefont{{Drossel}}},
  \emph{\bibinfo{title}{Random Boolean networks}} (\bibinfo{publisher}{Wiley},
  \bibinfo{address}{Weinheim}, \bibinfo{year}{2008}), vol.~\bibinfo{volume}{1}
  of \emph{\bibinfo{series}{Reviews of Nonlinear Dynamics and Complexity}},
  chap.~\bibinfo{chapter}{3}, pp. \bibinfo{pages}{69--96}.

\bibitem[{\citenamefont{Correale et~al.}(2006)\citenamefont{Correale, Leone,
  Pagnani, Weigt, and Zecchina}}]{Correale:RBN}
\bibinfo{author}{\bibfnamefont{L.}~\bibnamefont{Correale}},
  \bibinfo{author}{\bibfnamefont{M.}~\bibnamefont{Leone}},
  \bibinfo{author}{\bibfnamefont{A.}~\bibnamefont{Pagnani}},
  \bibinfo{author}{\bibfnamefont{M.}~\bibnamefont{Weigt}}, \bibnamefont{and}
  \bibinfo{author}{\bibfnamefont{R.}~\bibnamefont{Zecchina}},
  \bibinfo{journal}{Phys. Rev. Lett.} \textbf{\bibinfo{volume}{96}},
  \bibinfo{pages}{018101} (\bibinfo{year}{2006}).

\bibitem[{\citenamefont{Leone et~al.}()\citenamefont{Leone, Pagnani, Parisi,
  and Zagordi}}]{Leone:RBN}
\bibinfo{author}{\bibfnamefont{M.}~\bibnamefont{Leone}},
  \bibinfo{author}{\bibfnamefont{A.}~\bibnamefont{Pagnani}},
  \bibinfo{author}{\bibfnamefont{G.}~\bibnamefont{Parisi}}, \bibnamefont{and}
  \bibinfo{author}{\bibfnamefont{O.}~\bibnamefont{Zagordi}},
  \bibinfo{journal}{Journal of Statistical Mechanics: Theory and Experiment}
  \textbf{\bibinfo{volume}{2006}}, \bibinfo{pages}{P12012}.

\bibitem[{\citenamefont{{Derrida} and {Pomeau}}(1986)}]{Derrida}
\bibinfo{author}{\bibfnamefont{B.}~\bibnamefont{{Derrida}}} \bibnamefont{and}
  \bibinfo{author}{\bibfnamefont{Y.}~\bibnamefont{{Pomeau}}},
  \bibinfo{journal}{Europhys. Lett.} \textbf{\bibinfo{volume}{1}},
  \bibinfo{pages}{45} (\bibinfo{year}{1986}).

\bibitem[{\citenamefont{{Derrida} and {Weisbuch}}(1986)}]{DerridaAndW}
\bibinfo{author}{\bibfnamefont{B.}~\bibnamefont{{Derrida}}} \bibnamefont{and}
  \bibinfo{author}{\bibfnamefont{G.}~\bibnamefont{{Weisbuch}}},
  \bibinfo{journal}{J. Phys.} \textbf{\bibinfo{volume}{47}},
  \bibinfo{pages}{1297} (\bibinfo{year}{1986}).

\bibitem[{\citenamefont{Hilhorst and Nijmeijer}(1987)}]{Hilhorst}
\bibinfo{author}{\bibfnamefont{H.~J.} \bibnamefont{Hilhorst}} \bibnamefont{and}
  \bibinfo{author}{\bibfnamefont{M.}~\bibnamefont{Nijmeijer}},
  \bibinfo{journal}{J. Phys.} \textbf{\bibinfo{volume}{48}},
  \bibinfo{pages}{185} (\bibinfo{year}{1987}).

\bibitem[{\citenamefont{Kesseli et~al.}(2006)\citenamefont{Kesseli, R\"am\"o,
  and Yli-Harja}}]{Kesseli}
\bibinfo{author}{\bibfnamefont{J.}~\bibnamefont{Kesseli}},
  \bibinfo{author}{\bibfnamefont{P.}~\bibnamefont{R\"am\"o}}, \bibnamefont{and}
  \bibinfo{author}{\bibfnamefont{O.}~\bibnamefont{Yli-Harja}},
  \bibinfo{journal}{Phys. Rev. E.} \textbf{\bibinfo{volume}{74}},
  \bibinfo{pages}{046104} (\bibinfo{year}{2006}).

\bibitem[{\citenamefont{{Szejka} et~al.}(2008)\citenamefont{{Szejka},
  {Mihaljev}, and {Drossel}}}]{RTN}
\bibinfo{author}{\bibfnamefont{A.}~\bibnamefont{{Szejka}}},
  \bibinfo{author}{\bibfnamefont{T.}~\bibnamefont{{Mihaljev}}},
  \bibnamefont{and}
  \bibinfo{author}{\bibfnamefont{B.}~\bibnamefont{{Drossel}}},
  \bibinfo{journal}{New J. Phys.} \textbf{\bibinfo{volume}{10}},
  \bibinfo{pages}{063009} (\bibinfo{year}{2008}).

\bibitem[{\citenamefont{De~Dominicis}(1978)}]{dD}
\bibinfo{author}{\bibfnamefont{C.}~\bibnamefont{De~Dominicis}},
  \bibinfo{journal}{Phys. Rev. B.} \textbf{\bibinfo{volume}{18}},
  \bibinfo{pages}{4913} (\bibinfo{year}{1978}).

\bibitem[{\citenamefont{Mozeika and Saad}(2011)}]{Mozeika:BN}
\bibinfo{author}{\bibfnamefont{A.}~\bibnamefont{Mozeika}} \bibnamefont{and}
  \bibinfo{author}{\bibfnamefont{D.}~\bibnamefont{Saad}},
  \bibinfo{journal}{Phys. Rev. Lett.} \textbf{\bibinfo{volume}{106}},
  \bibinfo{pages}{214101} (\bibinfo{year}{2011}).

\bibitem[{\citenamefont{Mimura and Coolen}(2009)}]{MimuraAndCoolen}
\bibinfo{author}{\bibfnamefont{K.}~\bibnamefont{Mimura}} \bibnamefont{and}
  \bibinfo{author}{\bibfnamefont{A.~C.~C.} \bibnamefont{Coolen}},
  \bibinfo{journal}{J. Phys. A: Math. Theor.} \textbf{\bibinfo{volume}{42}},
  \bibinfo{pages}{415001} (\bibinfo{year}{2009}).

\bibitem[{\citenamefont{Neri and Boll\'{e}}(2009)}]{NeriAndBolle}
\bibinfo{author}{\bibfnamefont{I.}~\bibnamefont{Neri}} \bibnamefont{and}
  \bibinfo{author}{\bibfnamefont{D.}~\bibnamefont{Boll\'{e}}},
  \bibinfo{journal}{J. Stat. Mech. Theory Exp.}
  \textbf{\bibinfo{volume}{2009}}, \bibinfo{pages}{P08009}
  (\bibinfo{year}{2009}).

\bibitem[{\citenamefont{{Peixoto} and {Drossel}}(2009)}]{Noise}
\bibinfo{author}{\bibfnamefont{T.~P.} \bibnamefont{{Peixoto}}}
  \bibnamefont{and}
  \bibinfo{author}{\bibfnamefont{B.}~\bibnamefont{{Drossel}}},
  \bibinfo{journal}{Phys. Rev. E.} \textbf{\bibinfo{volume}{79}},
  \bibinfo{pages}{036108} (\bibinfo{year}{2009}).

\bibitem[{\citenamefont{{Derrida} and {Flyvbjerg}}(1987)}]{MinRBN}
\bibinfo{author}{\bibfnamefont{B.}~\bibnamefont{{Derrida}}} \bibnamefont{and}
  \bibinfo{author}{\bibfnamefont{H.}~\bibnamefont{{Flyvbjerg}}},
  \bibinfo{journal}{J. Phys. A: Math. Gen.} \textbf{\bibinfo{volume}{20}},
  \bibinfo{pages}{L1107} (\bibinfo{year}{1987}).

\bibitem[{\citenamefont{Mezard et~al.}(1987)\citenamefont{Mezard, Parisi, and
  Virasoro}}]{BOOK}
\bibinfo{author}{\bibfnamefont{M.}~\bibnamefont{Mezard}},
  \bibinfo{author}{\bibfnamefont{G.}~\bibnamefont{Parisi}}, \bibnamefont{and}
  \bibinfo{author}{\bibfnamefont{M.~A.} \bibnamefont{Virasoro}},
  \emph{\bibinfo{title}{Spin glass theory and beyond}}
  (\bibinfo{publisher}{World Scientific}, \bibinfo{address}{Singapore},
  \bibinfo{year}{1987}).

\bibitem[{\citenamefont{{Kree} and {Zippelius}}(1987)}]{adnn}
\bibinfo{author}{\bibfnamefont{R.}~\bibnamefont{{Kree}}} \bibnamefont{and}
  \bibinfo{author}{\bibfnamefont{A.}~\bibnamefont{{Zippelius}}},
  \bibinfo{journal}{Phys. Rev. A.} \textbf{\bibinfo{volume}{36}},
  \bibinfo{pages}{4421} (\bibinfo{year}{1987}).

\bibitem[{\citenamefont{McCann and Pippenger}(2008)}]{Pippenger}
\bibinfo{author}{\bibfnamefont{M.}~\bibnamefont{McCann}} \bibnamefont{and}
  \bibinfo{author}{\bibfnamefont{N.}~\bibnamefont{Pippenger}},
  \bibinfo{journal}{J. Comput. Syst. Sci.} \textbf{\bibinfo{volume}{74}},
  \bibinfo{pages}{910 } (\bibinfo{year}{2008}).

\bibitem[{\citenamefont{Evans and Schulman}(2003)}]{Evans:MTNK}
\bibinfo{author}{\bibfnamefont{W.}~\bibnamefont{Evans}} \bibnamefont{and}
  \bibinfo{author}{\bibfnamefont{L.}~\bibnamefont{Schulman}},
  \bibinfo{journal}{IEEE Trans. Inf. Theory} \textbf{\bibinfo{volume}{49}},
  \bibinfo{pages}{3094} (\bibinfo{year}{2003}).

\bibitem[{\citenamefont{Mozeika et~al.}(2009)\citenamefont{Mozeika, Saad, and
  Raymond}}]{NoisyPRL}
\bibinfo{author}{\bibfnamefont{A.}~\bibnamefont{Mozeika}},
  \bibinfo{author}{\bibfnamefont{D.}~\bibnamefont{Saad}}, \bibnamefont{and}
  \bibinfo{author}{\bibfnamefont{J.}~\bibnamefont{Raymond}},
  \bibinfo{journal}{Phys. Rev. Lett.} \textbf{\bibinfo{volume}{103}},
  \bibinfo{pages}{248701} (\bibinfo{year}{2009}).

\bibitem[{\citenamefont{{Peixoto}}(2010)}]{Peixoto}
\bibinfo{author}{\bibfnamefont{T.}~\bibnamefont{{Peixoto}}},
  \bibinfo{journal}{Phys. Rev. Lett.} \textbf{\bibinfo{volume}{104}},
  \bibinfo{pages}{048701} (\bibinfo{year}{2010}).

\bibitem[{\citenamefont{Von~Neumann}(1956)}]{VonNeumann}
\bibinfo{author}{\bibfnamefont{J.}~\bibnamefont{Von~Neumann}},
  \emph{\bibinfo{title}{Probabilistic logics and the synthesis of reliable
  organisms from unreliable components}} (\bibinfo{publisher}{Princeton
  University Press}, \bibinfo{address}{Princeton, NJ}, \bibinfo{year}{1956}),
  pp. \bibinfo{pages}{43--98}.

\bibitem[{\citenamefont{Coppersmith et~al.}(2001)\citenamefont{Coppersmith,
  Kadanoff, and Zhang}}]{Revers}
\bibinfo{author}{\bibfnamefont{S.~N.} \bibnamefont{Coppersmith}},
  \bibinfo{author}{\bibfnamefont{L.~P.} \bibnamefont{Kadanoff}},
  \bibnamefont{and} \bibinfo{author}{\bibfnamefont{Z.}~\bibnamefont{Zhang}},
  \bibinfo{journal}{Physica D} \textbf{\bibinfo{volume}{149}},
  \bibinfo{pages}{11 } (\bibinfo{year}{2001}).

\bibitem[{\citenamefont{de~Sales et~al.}(1997)\citenamefont{de~Sales, Martins,
  and Stariolo}}]{DeSales}
\bibinfo{author}{\bibfnamefont{J.~A.} \bibnamefont{de~Sales}},
  \bibinfo{author}{\bibfnamefont{M.~L.} \bibnamefont{Martins}},
  \bibnamefont{and} \bibinfo{author}{\bibfnamefont{D.~A.}
  \bibnamefont{Stariolo}}, \bibinfo{journal}{Phys. Rev. E.}
  \textbf{\bibinfo{volume}{55}}, \bibinfo{pages}{3262} (\bibinfo{year}{1997}).

\bibitem[{\citenamefont{Li et~al.}(2004)\citenamefont{Li, Long, Lu, Ouyang, and
  Tang}}]{Li}
\bibinfo{author}{\bibfnamefont{F.}~\bibnamefont{Li}},
  \bibinfo{author}{\bibfnamefont{T.}~\bibnamefont{Long}},
  \bibinfo{author}{\bibfnamefont{Y.}~\bibnamefont{Lu}},
  \bibinfo{author}{\bibfnamefont{Q.}~\bibnamefont{Ouyang}}, \bibnamefont{and}
  \bibinfo{author}{\bibfnamefont{C.}~\bibnamefont{Tang}},
  \bibinfo{journal}{Proc. Nat. Acad. Sci. U.S.A.}
  \textbf{\bibinfo{volume}{101}}, \bibinfo{pages}{4781} (\bibinfo{year}{2004}).

\bibitem[{\citenamefont{Heckel et~al.}(2010)\citenamefont{Heckel, Schober, and
  Bossert}}]{Heckel}
\bibinfo{author}{\bibfnamefont{R.}~\bibnamefont{Heckel}},
  \bibinfo{author}{\bibfnamefont{S.}~\bibnamefont{Schober}}, \bibnamefont{and}
  \bibinfo{author}{\bibfnamefont{M.}~\bibnamefont{Bossert}}, in
  \emph{\bibinfo{booktitle}{Source and Channel Coding (SCC), 2010 International
  ITG Conference on}} (\bibinfo{year}{2010}), pp. \bibinfo{pages}{1 --6}.

\bibitem[{\citenamefont{{Za{\~n}udo} et~al.}(2010)\citenamefont{{Za{\~n}udo},
  {Aldana}, and {Mart{\'{\i}}nez-Mekler}}}]{Zanudo}
\bibinfo{author}{\bibfnamefont{J.~G.~T.} \bibnamefont{{Za{\~n}udo}}},
  \bibinfo{author}{\bibfnamefont{M.}~\bibnamefont{{Aldana}}}, \bibnamefont{and}
  \bibinfo{author}{\bibfnamefont{G.}~\bibnamefont{{Mart{\'{\i}}nez-Mekler}}},
  \bibinfo{journal}{ArXiv e-prints}  (\bibinfo{year}{2010}),
  \eprint{1011.3848}.

\bibitem[{\citenamefont{Aldana and Cluzel}(2003)}]{scaleFreeBN}
\bibinfo{author}{\bibfnamefont{M.}~\bibnamefont{Aldana}} \bibnamefont{and}
  \bibinfo{author}{\bibfnamefont{P.}~\bibnamefont{Cluzel}},
  \textbf{\bibinfo{volume}{100}}, \bibinfo{pages}{8710} (\bibinfo{year}{2003}).

\bibitem[{\citenamefont{{Mozeika} et~al.}(2010)\citenamefont{{Mozeika}, {Saad},
  and {Raymond}}}]{NoisyPRE}
\bibinfo{author}{\bibfnamefont{A.}~\bibnamefont{{Mozeika}}},
  \bibinfo{author}{\bibfnamefont{D.}~\bibnamefont{{Saad}}}, \bibnamefont{and}
  \bibinfo{author}{\bibfnamefont{J.}~\bibnamefont{{Raymond}}},
  \bibinfo{journal}{Phys. Rev. E.} \textbf{\bibinfo{volume}{82}},
  \bibinfo{pages}{041112} (\bibinfo{year}{2010}).

\end{thebibliography}

\appendix
\section{Computation of the disorder averages\label{section:disorderAverage}}
In this section, we outline the calculation which takes us from the definition of  generating functional (\ref{eq:GF}) to the saddle-point integral (\ref{eq:GF-recurr-sp}). The starting point of this calculation is the generating functional
\begin{eqnarray}
\Gamma[\vecPsi^1;\vecPsi^2]&=&\sum_{\{S_{i}^{1}(t),S_{i}^{2}(t)\}}\Prob(\vecSpin^1(0),\vecSpin^2(0))\prod_{t\!=\!0}^{t_{max}\!-\!1}\!
\prod_{i=1}^N\prod_{\gamma =1}^{2}\frac{\rme^{\beta S_{i}^\gamma(t\!+\!1)h_{i}^{\gamma}(\vecSpin^\gamma(t))}}{2\cosh[\beta  h_{i}^{\gamma}(\vecSpin^\gamma(t))]}\label{eq:GF1}\\
&&\times\exp\left[-\rmi\sum_{t\!=\!0}^{t_{max}}\sum_{i\!=\!1}^N\sum_{\gamma\!=\!1}^2\psi_{i}^{\gamma }(t) S_{i}^{\gamma }(t)\right],
\nonumber
\end{eqnarray}
where in the above we have defined the field-variables $h_{i}^\gamma(\vecSpin^{\gamma}(t))=\sum_{i_1,\ldots,i_k}^N A_{i_1,\ldots,i_k}^{i} \alpha_i(S_{i_1}^{\gamma }(t),\ldots,S_{i_k}^{\gamma }(t))$. Removing these fields from equation~(\ref{eq:GF1}) via the integral representations of unity
\begin{eqnarray}
\prod_{t\!=\!0}^{t_{max}\!-\!1}\prod_{i=1}^{N}\prod_{\gamma =1}^{2}\left\{\int\frac{\mathrm d  h_{i}^{\gamma}(t)\;\mathrm d  \hat{h}_{i}^{\gamma }(t)}{2\pi}\;\rme^{\rmi  \hat{h}_{i}^{\gamma}(t)[ h_{i}^{\gamma}(t)- h_{i}^{\gamma}(\vecSpin^{\gamma}(t))]}\right\}\!=\!1\label{def:unity}
\end{eqnarray}
gives
\begin{eqnarray}
\Gamma[\vecPsi^1;\vecPsi^2]&=&\sum_{\{S_{i}^{1}(t),S_{i}^{2}(t)\}}\Prob(\vecSpin^1(0),\vecSpin^2(0))\exp\left[-\rmi\sum_{t\!=\!0}^{t_{max}}\sum_{i\!=\!1}^N\sum_{\gamma\!=\!1}^2\psi_{i}^{\gamma }(t) S_{i}^{\gamma }(t)\right]\label{eq:GF2}\\
&&\times\prod_{t=0}^{t_{max}\!-\!1}\!
\prod_{i=1}^N\prod_{\gamma =1}^{2}\left\{\int\frac{\mathrm d  h_{i}^{\gamma}(t)\;\mathrm d  \hat{h}_{i}^{\gamma }(t)}{2\pi}\;\rme^{\rmi  \hat{h}_{i}^{\gamma}(t) h_{i}^{\gamma}(t)}\frac{\rme^{\beta S_{i}^\gamma(t\!+\!1)h_{i}^{\gamma}(t)}}{2\cosh[\beta  h_{i}^{\gamma}(t)]}\right\}
\nonumber\\
&&\times\prod_{t=0}^{t_{max}\!-\!1}\!
\prod_{i=1}^N\prod_{\gamma =1}^{2}\rme^{-\rmi  \hat{h}_{i}^{\gamma }(t) \sum_{i_1,\ldots,i_k}^N A_{i_1,\ldots,i_k}^{i} \alpha_i(S_{i_1}^{\gamma}(t),\ldots,S_{i_k}^{\gamma}(t))}\label{eq:disorder}.
\end{eqnarray}

Now we can average out the disorder in (\ref{eq:disorder})
\begin{eqnarray}
&&\overline{ \prod_{t=0}^{t_{max}\!-\!1}\!
\prod_{i=1}^N\prod_{\gamma =1}^{2}\prod_{i_1,\ldots,i_k}^N\rme^{-\rmi  \hat{h}_{i}^{\gamma }(t)  A_{i_1,\ldots,i_k}^{i} \alpha_i(S_{i_1}^{\gamma}(t),\ldots,S_{i_k}^{\gamma}(t))}}\label{eq:aver-recurr}\\
&&=\frac{1}{Z_A}\!
\prod_{i=1}^N
\prod_{\setI\subseteq\setN}\left\{\sum_{A_{\setI}^{i}}\left[\frac{1}{N^k}\delta_{A_{\setI}^{i};1}+(1-\frac{1}{N^k})\delta_{A_{\setI}^{i};0}\right]\right\}\delta\left[1; \sum_{\setI^\prime\subseteq\setN} A_{\setI^\prime}^{i}\right]\nonumber\\
&&\times\sum_{\alpha_i}\Prob(\alpha_{i})\prod_{i_1,\ldots,i_k}^N\rme^{-\rmi\sum_{t=0}^{t_{max}\!-\!1}\!
\sum_{\gamma =1}^{2}  \hat{h}_{i}^{\gamma }(t)  A_{i_1,\ldots,i_k}^{i} \alpha_i(S_{i_1}^{\gamma}(t),\ldots,S_{i_k}^{\gamma}(t))}\nonumber\\
&&=\frac{1}{Z_A}\!
\prod_{i=1}^N
\prod_{\setI\subseteq\setN}\left\{\sum_{A_{\setI}^{i}}\left[\frac{1}{N^k}\delta_{A_{\setI}^{i};1}+(1-\frac{1}{N^k})\delta_{A_{\setI}^{i};0}\right]\right\}\delta\left[1; \sum_{\setI^\prime\subseteq\setN} A_{\setI^\prime}^{i}\right]\nonumber\\
&&\times\prod_{i_1,\ldots,i_k}^N\overline{\rme^{-\rmi\sum_{t=0}^{t_{max}\!-\!1}\!
\sum_{\gamma =1}^{2}  \hat{h}_{i}^{\gamma }(t)  A_{i_1,\ldots,i_k}^{i} \alpha_i(S_{i_1}^{\gamma}(t),\ldots,S_{i_k}^{\gamma}(t))}}^{\;\alpha_i}\nonumber\\
&&=\frac{1}{Z_A}\!\left\{\prod_{i=1}^N\int_{-\pi}^{\pi}\frac{\mathrm d \omega_i}{2\pi}\rme^{\rmi\omega_i}\right\}\nonumber\\
&&\times\exp\!\left[\frac{1}{N^k}\!\!\sum_{i,i_1,\ldots,i_k}^{N}\!\!\overline{\left(\rme^{-\rmi\sum_{t=0}^{t_{max}\!-\!1}\!
\sum_{\gamma =1}^{2}  \hat{h}_{i}^{\gamma }(t) \alpha_i(S_{i_1}^{\gamma}(t),\ldots,S_{i_k}^{\gamma}(t))-\rmi\omega_i}\! -\!1\right)}^{\; \alpha}\!\!+O(N^{-k+1})\right]\nonumber
\end{eqnarray}
and use this result in the generating functional~(\ref{eq:GF2}) to obtain
\begin{eqnarray}
&&\overline{\Gamma[\vecPsi^1;\vecPsi^2]}\label{eq:GF-recurr-1}\\
&&=\sum_{\{S_{i}^{1}(t),S_{i}^{2}(t)\}}\Prob(\vecSpin^1(0),\vecSpin^2(0))\exp\left[-\rmi\sum_{t\!=\!0}^{t_{max}}\sum_{i\!=\!1}^N\sum_{\gamma\!=\!1}^2\psi_{i}^{\gamma }(t) S_{i}^{\gamma }(t)\right]\nonumber\\
&\times&\prod_{t=0}^{t_{max}\!-\!1}\!
\prod_{i=1}^N\prod_{\gamma =1}^{2}\left\{\int\frac{\mathrm d  h_{i}^{\gamma}(t)\mathrm d  \hat{h}_{i}^{\gamma }(t)}{2\pi}\;\rme^{\rmi  \hat{h}_{i}^{\gamma}(t) h_{i}^{\gamma}(t)}\frac{\rme^{\beta S_{i}^\gamma(t\!+\!1)h_{i}^{\gamma}(t)}}{2\cosh[\beta  h_{i}^{\gamma}(t)]}\right\}
\nonumber\\
&\times&\frac{1}{Z_A}\!\left\{\prod_{i=1}^N\int_{-\pi}^{\pi}\frac{\mathrm d \omega_i}{2\pi}\rme^{\rmi\omega_i}\right\}\exp\!\Big[N\int\{\mathrm d\mathbf{\hat{h}}(t)\}\int_{-\pi}^{\pi}\mathrm d\omega\nonumber\\
&\times&\frac{1}{N}\sum_{i=1}^{N}\left\{\delta(\omega-\omega_i)\prod_{t=0}^{t_{max}\!-\!1}\delta(\mathbf{\hat{h}}(t)-\mathbf{\hat{h}}_i(t))\right\}\nonumber\\
&\times&\sum_{\{\vecSpin_j(t)\}}\frac{1}{N^k}\!\sum_{i_1,\ldots,i_k}^{N}\left\{\prod_{t=0}^{t_{max}\!-\!1}\prod_{j=1}^{k}\delta[\vecSpin_j(t);\vecSpin_{i_j}(t)]\right\}\nonumber\\
&\times&\overline{\left(\rme^{-\rmi\sum_{t=0}^{t_{max}\!-\!1}\!
\sum_{\gamma =1}^{2}  \hat{h}^{\gamma }(t)\;\alpha(S_{1}^{\gamma}(t),\ldots,S_{k}^{\gamma}(t))-\rmi\omega} -1\right)}^{\;\alpha}+O(N^{-k+1})\Big]\nonumber,
\end{eqnarray}
where in the above we have defined the vectors $\mathbf{\hat{h}}_i(t)=(\hat{h}_{i}^1(t),\hat{h}_{i}^2(t))$ and $\vecSpin_{i}(t)=(S_{i}^{1}(t),S_{i}^{2}(t))$. Using the unity representations
 \begin{eqnarray}
&&\int\{\mathrm d P \mathrm d\hat{P}\}\rme^{\rmi N\!\sum_{\{\mathbf{S}(t)\}}\hat{P}(\{\mathbf{S}(t)\})[P(\{\mathbf{S}(t)\})-\frac{1}{N}\sum_{i=1}^N\prod_{t=0}^{t_{max}\!-\!1}\delta_{\mathbf{S}(t);\mathbf{S}_{i}(t)}]}=1\\
&&\int\{\mathrm d\Omega \mathrm d\hat{\Omega}\}\rme^{\rmi N\! \int\{\rmd \mathbf{\hat{h}}(t)\}\rmd \omega\hat{\Omega}(\{\mathbf{\hat{h}}(t)\},\omega)[\Omega(\{\mathbf{\hat{h}}(t)\},\omega)\!-\!\frac{1}{N}\sum_{i=1}^N\left[\prod_{t=0}^{t_{max}\!-\!1}\delta(\mathbf{\hat{h}}(t)\!-\!\mathbf{\hat{h}}_i(t))\right]\delta(\omega\!-\!\omega_i)]}=1\nonumber
\end{eqnarray}
gives
\begin{eqnarray}
&&\overline{\Gamma[\vecPsi^1;\vecPsi^2]}=\int\{\mathrm d P \mathrm d\hat{P}\mathrm d\Omega \mathrm d\hat{\Omega}\}\nonumber\label{eq:GF-recurr-2}\\
&&\times\exp N\Big[\rmi\sum_{\{\vecSpin (t)\}}\hat{P}(\{\vecSpin (t)\})P(\{\vecSpin (t)\})+\rmi\int\{\rmd \mathbf{\hat{h}}(t)\}\;\rmd \omega\;\hat{\Omega}(\{\mathbf{\hat{h}}(t)\},\omega)\;\Omega(\{\mathbf{\hat{h}}(t)\},\omega)\nonumber\\
&&+\int\{\mathrm d\mathbf{\hat{h}}(t)\}\;\mathrm d\omega\;\Omega(\{\mathbf{\hat{h}}(t)\},\omega)\sum_{\{\mathbf{S}_j(t)\}}\left\{\prod_{j=1}^k P(\{\vecSpin_j(t)\})\right\}
\\
&&\times\overline{\left(\rme^{-\rmi\sum_{t=0}^{t_{max}\!-\!1}\!
\sum_{\gamma =1}^{2}  \hat{h}^{\gamma }(t)\;\alpha(S_{1}^{\gamma}(t),\ldots,S_{k}^{\gamma}(t))-\rmi\omega} -1\right)}^{\;\alpha}-\frac{1}{N}\log Z_A\Big]\nonumber\\
&&\times\sum_{\{S_{i}^{1}(t),S_{i}^{2}(t)\}}\Prob(\vecSpin^1(0),\vecSpin^2(0))\exp\left[-\rmi\sum_{t\!=\!0}^{t_{max}}\sum_{i\!=\!1}^N\sum_{\gamma\!=\!1}^2\psi_{i}^{\gamma }(t) S_{i}^{\gamma }(t)\right]\nonumber\\
&&\times\prod_{t=0}^{t_{max}\!-\!1}\!
\prod_{i=1}^N\prod_{\gamma =1}^{2}\left\{\int\frac{\mathrm d  h_{i}^{\gamma}(t)\mathrm d  \hat{h}_{i}^{\gamma }(t)}{2\pi}\;\rme^{\rmi  \hat{h}_{i}^{\gamma}(t) h_{i}^{\gamma}(t)}\frac{\rme^{\beta S_{i}^\gamma(t\!+\!1)h_{i}^{\gamma}(t)}}{2\cosh[\beta  h_{i}^{\gamma}(t)]}\right\}
\nonumber\\
&&\times\left\{\prod_{i=1}^N\int_{-\pi}^{\pi}\frac{\mathrm d \omega_i}{2\pi}\rme^{\rmi\omega_i}\right\}\rme^{-\rmi\sum_{i=1}^N\hat{P}(\{\vecSpin_i(t)\})-\rmi\sum_{i=1}^N\hat{\Omega}(\{\mathbf{\hat{h}}_i(t)\},\omega_i)}.
\nonumber
\end{eqnarray}
Equation~(\ref{eq:GF-recurr-2}) gives one the saddle-point integral~(\ref{eq:GF-recurr-sp}) if write its site-dependent part in the exponential form
\begin{eqnarray}
\exp\left[\sum_{i=1}^N\log\sum_{\{\vecSpin_i(t)\}}\int\{\mathrm d \mathbf{h}_i(t)\mathrm d \mathbf{\hat{h}}_i(t)\}\int_{-\pi}^{\pi}\frac{\mathrm d \omega_i}{2\pi}M[\{\vecSpin_i(t),\mathbf{h}_i(t)\}\vert\{\mathbf{\hat{h}}_i(t)\}, \omega_i,\{\psi_i^\gamma(t)\}]\right]\nonumber\label{def:Mpart},
\end{eqnarray}
where the definition
\begin{eqnarray}
&&M[\{\vecSpin_i(t),\mathbf{h}_i(t)\}\vert\{\mathbf{\hat{h}}_i(t)\}, \omega_i,\{\psi_i^\gamma(t)\}]\label{def:M-recurr}\\
&&=\Prob(S_i^1(0),S_i^2(0))\exp\left[-\rmi\sum_{t\!=\!0}^{t_{max}}\sum_{\gamma\!=\!1}^2\psi_{i}^{\gamma }(t) S_{i}^{\gamma }(t)\right]\nonumber\\
&&\times \prod_{t=0}^{t_{max}\!-\!1}\!
\prod_{\gamma =1}^{2}\left\{\;\rme^{\rmi  \hat{h}_{i}^{\gamma}(t) h_{i}^{\gamma}(t)}\frac{\rme^{\beta S_{i}^\gamma(t\!+\!1)h_{i}^{\gamma}(t)}}{2\cosh[\beta  h_{i}^{\gamma}(t)]}\right\}\nonumber\\
&&\times\rme^{-\rmi \hat{\Omega}(\{\mathbf{\hat{h}}_i(t)\},\omega_i)+\rmi\omega_i-\rmi\hat{P}(\{\vecSpin_i(t)\})}\nonumber
\end{eqnarray}
is used with the shorthand $\int\{\mathrm d \mathbf{h}_i(t)\;\mathrm d \mathbf{\hat{h}}_i(t)\}=\prod_{t=0}^{t_{max}\!-\!1}\!
\prod_{\gamma =1}^{2}\int\frac{\mathrm d  h_{i}^{\gamma}(t)\;\mathrm d  \hat{h}_{i}^{\gamma }(t)}{2\pi}$.

\section{Solution of the saddle-point problem\label{section:solOfSP}}

In this section, we show that the conjugate order-parameter function $\hat{P}(\{\vecSpin (t)\})$, which is governed by the equation (\ref{eq:SP2-recurr}), is a constant function. In order to do this, we first rewrite the (disorder-averaged) path-probability (\ref{eq:PathProb}) as follows

\begin{eqnarray}
&&\overline{\Prob[\{\vecSpin^1(t)\};\{\vecSpin^2(t)\}]}=\frac{1}{Z_A}\Prob(\vecSpin^1(0),\vecSpin^2(0))\label{eq:disorder-averaged-PathProb}\\
&&\times\prod_{i=1}^N\overline{\!\Bigg\{\prod_{t\!=\!0}^{t_{max}\!-\!1}\!
\prod_{\gamma =1}^{2}\frac{\rme^{\beta S_{i}^\gamma(t\!+\!1)\sum_{i_1,\ldots,i_k}^N A_{i_1,\ldots,i_k}^{i} \alpha_i(S_{i_1}^{\gamma }(t),\ldots,S_{i_k}^{\gamma }(t))}}{2\cosh[\beta\sum_{i_1,\ldots,i_k}^N A_{i_1,\ldots,i_k}^{i} \alpha_i(S_{i_1}^{\gamma }(t),\ldots,S_{i_k}^{\gamma }(t))]}}\nonumber\\
&&\times\overline{\delta\left[1;\!\!\sum_{\setI^\prime\subseteq\setN} A_{\setI^\prime}^{i}\right]\Bigg\}}^{\{A_{\setI}^{i}\}}\nonumber\\
&=&\frac{1}{Z_A}\prod_{i=1}^{N}\Bigg[\int\{\mathrm d \mathbf{h}_i(t)\;\mathrm d \mathbf{\hat{h}}_i(t)\}\int_{-\pi}^{\pi}\frac{\mathrm d \omega_i}{2\pi}\;\overline{\rme^{\rmi\omega_i(1- \sum_{i_1,\ldots,i_k}^NA_{i_1,\ldots,i_k}^{i})}\Prob(S_i^1(0),S_i^2(0))}\label{def:Xi}\\
&&\times\overline{\left\{\prod_{t\!=\!0}^{t_{max}\!-\!1}\!
\prod_{\gamma =1}^{2}\;\rme^{\rmi  \hat{h}_{i}^{\gamma}(t)[ h_{i}^{\gamma}(t)- \sum_{i_1,\ldots,i_k}^N A_{i_1,\ldots,i_k}^{i} \alpha_i(S_{i_1}^{\gamma }(t),\ldots,S_{i_k}^{\gamma }(t))]}\frac{\rme^{\beta S_{i}^\gamma(t\!+\!1)h_{i}^{\gamma}(t)}}{2\cosh[\beta h_{i}^{\gamma}(t)]}\right\}}^{\{A_{\setI}^{i}\}}\nonumber\Bigg]\\
&=&\frac{1}{Z_A}\prod_{i=1}^{N}\int\{\mathrm d \mathbf{h}_i(t)\;\mathrm d \mathbf{\hat{h}}_i(t)\}\int_{-\pi}^{\pi}\frac{\mathrm d \omega_i}{2\pi}
\;\overline{\Xi_i[\{\vecSpin_i(t),\mathbf{h}_i(t)\}\vert\{\mathbf{\hat{h}}_i(t)\}, \omega_i]}^{\{A_{\setI}^{i}\}},\nonumber
\end{eqnarray}
where in the above we have averaged out the connectivity disorder only, i.e. $\overline{(\cdots)}^{\{A_{\setI}^{i}\}}=\sum_{\{A_{\setI}^{i}\}}\prod_{\setI\subseteq\setN}\left\{\frac{1}{N^k}\delta_{A_{\setI}^{i};1}+(1\!-\!\frac{1}{N^k})\delta_{A_{\setI}^{i};0}\right\}(\cdots)$ and the definition of single-site measure $\Xi_i$ is clear from equation~(\ref{def:Xi}).

In the next step, one notes that the effect of the Fourier transform
\begin{eqnarray}
&&\int\!\!\{\mathrm d \mathbf{h}_i(t)\;\mathrm d \mathbf{\hat{h}}_i(t)\}\!\!\!\int_{-\pi}^{\pi}\!\!\!\frac{\mathrm d \omega_i}{2\pi}
\;\Xi_i[\{\vecSpin_i(t),\mathbf{h}_i(t)\}\vert\{\mathbf{\hat{h}}_i(t)\}, \omega_i]\label{def:Fourier}\\
&&\times\rme^{-\rmi\sum_{t\!=\!0}^{t_{max}\!-\!1}\!
\mathbf{\hat{h}}_{i}(t).\theta(t)-\rmi\omega_i}\nonumber\\
&=&\int\{\mathrm d \mathbf{h}_i(t)\;\mathrm d \mathbf{\hat{h}}_i(t)\}\int_{-\pi}^{\pi}\frac{\mathrm d \omega_i}{2\pi}\;\rme^{\rmi\omega_i(0- \sum_{i_1,\ldots,i_k}^NA_{i_1,\ldots,i_k}^{i})}\Prob(S_i^1(0),S_i^2(0))
\nonumber\\
&&\times\Bigg[\prod_{t\!=\!0}^{t_{max}\!-\!1}\!
\prod_{\gamma =1}^{2}\;\rme^{\rmi  \hat{h}_{i}^{\gamma}(t)[ h_{i}^{\gamma}(t)-\theta^\gamma(t)- \sum_{i_1,\ldots,i_k}^N A_{i_1,\ldots,i_k}^{i} \alpha_i(S_{i_1}^{\gamma }(t),\ldots,S_{i_k}^{\gamma }(t))]}\nonumber\\
&&\times\frac{\rme^{\beta S_{i}^\gamma(t\!+\!1)h_{i}^{\gamma}(t)}}{2\cosh[\beta h_{i}^{\gamma}(t)]}\Bigg]\nonumber\\
&=&\left\{\prod_{t\!=\!0}^{t_{max}\!-\!1}\!
\prod_{\gamma =1}^{2}\frac{\rme^{\beta S_{i}^\gamma(t\!+\!1)\theta^{\gamma}(t)}}{2\cosh[\beta \theta^{\gamma}(t)]}\right\}\delta\left[0;\sum_{i_1,\ldots,i_k}^NA_{i_1,\ldots,i_k}^{i}\right]\Prob(S_i^1(0),S_i^2(0)),\nonumber
\end{eqnarray}
on the function $\Xi_i$ is to replace the Boolean function $ \alpha_i$ on site $i$ with a constant function  $\theta^\gamma(t)\in\{-1,1\}$. With this in mind we can define the average site-perturbed path-probability
\begin{eqnarray}
&&\overline{\frac{1}{N}\sum_{i=1}^N\Prob[\{\vecSpin^1(t)\};\{\vecSpin^2(t)\}]_{\vert_{\alpha_i\rightarrow\theta}}}\label{eq:pertrub-PathProb}\\
&=&\frac{Z_A^{-1}}{N}\sum_{i=1}^N\Bigg\{\prod_{j\neq i}^{N}\int\{\mathrm d \mathbf{h}_j(t)\;\mathrm d \mathbf{\hat{h}}_j(t)\}\nonumber\\
&&\times\int_{-\pi}^{\pi}\frac{\mathrm d \omega_j}{2\pi}
\;\overline{\Xi_j[\{\vecSpin_j(t),\mathbf{h}_j(t)\}\vert\{\mathbf{\hat{h}}_j(t)\}, \omega_i]}^{\{A_{\setI}^{j}\}}\Bigg\}\nonumber\\
&&\times\int\!\!\{\mathrm d \mathbf{h}_i(t)\;\mathrm d \mathbf{\hat{h}}_i(t)\}\!\!\!\int_{-\pi}^{\pi}\!\!\!\frac{\mathrm d \omega_i}{2\pi}
\;\overline{\Xi_i[\{\vecSpin_i(t),\mathbf{h}_i(t)\}\vert\{\mathbf{\hat{h}}_i(t)\}, \omega_i]}^{\{A_{\setI}^{i}\}}\nonumber\\
&&\times\rme^{-\rmi\sum_{t\!=\!0}^{t_{max}\!-\!1}\!
\mathbf{\hat{h}}_{i}(t).\theta(t)-\rmi\omega_i}.\nonumber
\end{eqnarray}
To show that the object in~(\ref{eq:pertrub-PathProb}) indeed defines a probability measure we note that it is clearly a positive semi-definite and the normalization of ~(\ref{eq:pertrub-PathProb}) can be established as follows
\begin{eqnarray}
&&\sum_{\{\mathbf{S}^1(t),\mathbf{S}^2(t)\}}\overline{\frac{1}{N}\sum_{i=1}^N\Prob[\{\vecSpin^1(t)\};\{\vecSpin^2(t)\}]_{\vert_{\alpha_i\rightarrow\theta}}}\label{eq:unity-1}\\
&=&\frac{Z_A^{-1}}{N}\sum_{i=1}^N\Bigg\{\prod_{j\neq i}^{N}\sum_{\{\mathbf{S}_j(t)\}}\int\{\mathrm d \mathbf{h}_j(t)\;\mathrm d \mathbf{\hat{h}}_j(t)\}\nonumber\\
&&\times\int_{-\pi}^{\pi}\frac{\mathrm d \omega_j}{2\pi}
\;\overline{\Xi_j[\{\vecSpin_j(t),\mathbf{h}_j(t)\}\vert\{\mathbf{\hat{h}}_j(t)\}, \omega_i]}^{\{A_{\setI}^{j}\}}\Bigg\}\nonumber\\
&&\times\sum_{\{\mathbf{S}_i(t)\}}\int\!\!\{\mathrm d \mathbf{h}_i(t)\;\mathrm d \mathbf{\hat{h}}_i(t)\}\!\!\!\int_{-\pi}^{\pi}\!\!\!\frac{\mathrm d \omega_i}{2\pi}
\;\overline{\Xi_i[\{\vecSpin_i(t),\mathbf{h}_i(t)\}\vert\{\mathbf{\hat{h}}_i(t)\}, \omega_i]}^{\{A_{\setI}^{i}\}}\nonumber\\
&&\times\rme^{-\rmi\sum_{t\!=\!0}^{t_{max}\!-\!1}\!
\mathbf{\hat{h}}_{i}(t).\theta(t)-\rmi\omega_i}\nonumber\\
&=&\frac{1}{N}\sum_{i=1}^N\prod_{j\neq i}^{N}\overline{\delta\left[1;\!\!\sum_{\setI^\prime\subseteq\setN} A_{\setI^\prime}^{j}\right]}^{\{A_{\setI}^{j}\}}
\overline{\delta\left[0;\!\!\sum_{\setI\subseteq\setN} A_{\setI}^{i}\right]}^{\{A_{\setI}^{i}\}}
/Z_A\nonumber\\
%
%
&=&\frac{1}{N}\sum_{i=1}^N\prod_{j\neq i}^{N}\sum_{\{A_{\setI}^{j}\}}\prod_{\setI\subseteq\setN}%
(1\!-\!\frac{1}{N^k})\exp\left[A_{\setI}^{j}\log\left(\frac{1}{N^k}(1\!-\!\frac{1}{N^k})^{-1}\right)\right]\delta\left[1;\!\!\sum_{\setI^\prime\subseteq\setN} A_{\setI^\prime}^{j}\right]\nonumber\\
&&\times\sum_{\{A_{\setI}^{i}\}}\prod_{\setI\subseteq\setN}(1\!-\!\frac{1}{N^k})\exp\left[A_{\setI}^{i}\log\left(\frac{1}{N^k}(1\!-\!\frac{1}{N^k})^{-1}\right)\right]\delta\left[0;\!\!\sum_{\setI\subseteq\setN} A_{\setI}^{i}\right]/Z_A\nonumber\\
%
 %
&=&\frac{1}{Z_A}\left\{(1\!-\!\frac{1}{N^k})^{N^k-1}\right\}^{N-1}(1\!-\!\frac{1}{N^k})^{N^k}\approx\frac{\rme^{-N}}{Z_A}\nonumber
\end{eqnarray}
Since for $N\rightarrow\infty$ the inverse of $Z_A$ in (\ref{eq:unity-1}) grows as $\rme^{N}$, we conclude that $\lim_{N\rightarrow\infty}\sum_{\{\mathbf{S}^1(t),\mathbf{S}^2(t)\}}\overline{\frac{1}{N}\sum_{i=1}^N\Prob[\{\vecSpin^1(t)\};\{\vecSpin^2(t)\}]_{\vert_{\alpha_i\rightarrow\theta}}}=1$.
Alternatively, the calculation in~(\ref{eq:unity-1}) can be carried out differently using results from the appendix~\ref{section:disorderAverage}:
\begin{eqnarray}
&&\sum_{\{\mathbf{S}^1(t),\mathbf{S}^2(t)\}}\overline{\frac{1}{N}\sum_{i=1}^N\Prob[\{\vecSpin^1(t)\};\{\vecSpin^2(t)\}]_{\vert_{\alpha_i\rightarrow\theta}}}\label{eq:unity-2}\\
&=&\frac{Z_A^{-1}}{N}\sum_{i=1}^N\Bigg\{\prod_{j\neq i}^{N}\sum_{\{\mathbf{S}_j(t)\}}\int\{\mathrm d \mathbf{h}_j(t)\;\mathrm d \mathbf{\hat{h}}_j(t)\}\nonumber\\
&&\times\int_{-\pi}^{\pi}\frac{\mathrm d \omega_j}{2\pi}
\;\overline{\Xi_j[\{\vecSpin_j(t),\mathbf{h}_j(t)\}\vert\{\mathbf{\hat{h}}_j(t)\}, \omega_i]}^{\{A_{\setI}^{j},\alpha_j\}}\Bigg\}\nonumber\\
&&\times\sum_{\{\mathbf{S}_i(t)\}}\int\!\!\{\mathrm d \mathbf{h}_i(t)\;\mathrm d \mathbf{\hat{h}}_i(t)\}\!\!\!\int_{-\pi}^{\pi}\!\!\!\frac{\mathrm d \omega_i}{2\pi}
\;\overline{\Xi_i[\{\vecSpin_i(t),\mathbf{h}_i(t)\}\vert\{\mathbf{\hat{h}}_i(t)\}, \omega_i]}^{\{A_{\setI}^{i},\alpha_i\}}\nonumber\\
&&\times\rme^{-\rmi\sum_{t\!=\!0}^{t_{max}\!-\!1}\!
\mathbf{\hat{h}}_{i}(t).\theta(t)-\rmi\omega_i}\nonumber\\
&=&\frac{1}{N}\sum_{i=1}^N\frac{\int\{\mathrm d P \mathrm d\hat{P}\mathrm d\Omega \mathrm d\hat{\Omega}\}\;\rme^{N\Psi[\{P,\hat{P},\Omega,\hat{\Omega}\}]}}{\int\{\mathrm d P^\prime \mathrm d\hat{P}^\prime\mathrm d\Omega^\prime \mathrm d\hat{\Omega}^\prime\}\;\rme^{N\Psi[\{P^\prime,\hat{P}^\prime,\Omega^\prime,\hat{\Omega}^\prime\}]}}\left\langle \rme^{-\rmi\sum_{t\!=\!0}^{t_{max}\!-\!1}\!
\mathbf{\hat{h}}_{i}(t).\theta(t)-\rmi\omega_i}\right\rangle_{M_i},\nonumber
\end{eqnarray}
where in the above $M_i$-average is generated by the site-dependent version of~(\ref{def:saddle-recurr}). Computing the integrals in~(\ref{eq:unity-2}) by the saddle-point method we find that  $\lim_{N\rightarrow\infty}\sum_{\{\mathbf{S}^1(t),\mathbf{S}^2(t)\}}\overline{\frac{1}{N}\sum_{i=1}^N\Prob[\{\vecSpin^1(t)\};\{\vecSpin^2(t)\}]_{\vert_{\alpha_i\rightarrow\theta}}}=\int\{\mathrm d\mathbf{\hat{h}}(t)\}\;\mathrm d\omega\;\Omega(\{\mathbf{\hat{h}}(t)\},\omega)\;\rme^{-\rmi\sum_{t\!=\!0}^{t_{max}\!-\!1}\!
\mathbf{\hat{h}}(t).\theta(t)-\rmi\omega}$, where the order-parameter function $\Omega$ is defined in~(\ref{eq:SP3-recurr}), but according to the calculation in (\ref{eq:unity-1}) this also equals unity. Thus, by using this result in the saddle-point equation~(\ref{eq:SP2-recurr}), we find that the equality $\hat{P}(\{\vecSpin (t)\})=\rmi k$ holds.

\end{document}